\documentclass{article}
\author{Dan Betea\thanks{Institute for Applied Mathematics, University of Bonn, D-53315 Bonn, Germany, \texttt{dan.betea@gmail.com}}
  \and
  J\'er\'emie Bouttier\thanks{Institut de Physique Th\'eorique, Universit\'e Paris-Saclay, CEA, CNRS, F-91191 Gif-sur-Yvette, \texttt{jeremie.bouttier@ipht.fr} \protect\\
    \indent \ \,
    Univ Lyon, Ens de Lyon, Univ Claude Bernard, CNRS, Laboratoire de Physique, F-69342 Lyon, France}}
\title{The periodic Schur process and free fermions at finite temperature}

\usepackage{graphicx, graphics, amsmath, amsthm, amssymb, xcolor,mathtools,bbm,mathabx}
\usepackage{hyperref}
\usepackage{tikz}
\usetikzlibrary{calc,decorations.pathmorphing,decorations.markings, decorations.pathreplacing,patterns,shapes,arrows}
\usepackage[margin=1.0in]{geometry}
\usepackage{subcaption}
\usepackage{stmaryrd}
\newtheorem{thm}{Theorem}
\newtheorem{prop}[thm]{Proposition}

\newtheorem{lem}[thm]{Lemma}

\theoremstyle{definition}
\newtheorem{definition}[thm]{Definition}
\theoremstyle{remark}
\newtheorem{rem}[thm]{Remark}
\numberwithin{equation}{section}
\numberwithin{thm}{section}

\DeclareMathAlphabet{\mathpzc}{OT1}{pzc}{m}{it}

\newcommand{\R}{\mathbb{R}}
\newcommand{\Z}{\mathbb{Z}}
\newcommand{\N}{\mathbb{N}}
\newcommand{\C}{\mathbb{C}}
\newcommand{\Gami}{\Gamma_{-}}
\newcommand{\Gapl}{\Gamma_{+}}
\newcommand{\Gatpl}{\Gamma'_{+}}
\newcommand{\Gatmi}{\Gamma'_{-}}

\DeclareMathOperator{\tr}{tr}
\DeclareMathOperator*{\pf}{pf}
\newcommand{\Par}{\mathcal{P}}
\newcommand{\SPar}{\mathcal{SP}}
\newcommand{\F}{\mathcal{F}}
\newcommand{\NF}{\mathcal{NF}}

\DeclareMathOperator{\Prob}{Prob}
\DeclareMathOperator{\Ad}{Ad}
\DeclareMathOperator{\Ai}{Ai}
\DeclareMathOperator{\ex}{ex}
\DeclareMathOperator{\arsinh}{arsinh}
\newcommand{\bra}[1]{\langle #1 |}
\newcommand{\ket}[1]{| #1 \rangle}
\newcommand{\evuc}[1]{\langle #1 \rangle_u^{(0)}}
\newcommand{\evut}[1]{\langle #1 \rangle_{u,t}}
\newcommand{\vv}{| \emptyset \rangle}
\newcommand{\vcv}{\langle \emptyset |}
\newcommand{\Tt}{T^{(\vartheta)}}

\newcommand{\im }{\mathrm{i}}  % this is the imaginary unit

\newcommand{\usf}{\mathsf{u}}
\newcommand{\vf}{\varphi}

\begin{document}

\maketitle

\begin{abstract}
  We revisit the periodic Schur process introduced by Borodin in
  2007. Our contribution is threefold. First, we provide a new simpler
  derivation of its correlation functions via the free fermion
  formalism. In particular, we shall see that the process becomes
  determinantal by passing to the grand canonical ensemble, which
  gives a physical explanation to Borodin's ``shift-mixing'' trick.
  Second, we consider the edge scaling limit in the simplest
  nontrivial case, corresponding to a deformation of the poissonized
  Plancherel measure on partitions. We show that the edge behavior is
  described, in a certain crossover regime different from that for the
  bulk, by the universal \emph{finite-temperature Airy kernel}, which
  was previously encountered by Johansson and Le Doussal \emph{et
    al.}\ in other models, and whose extreme value statistics
  interpolates between the Tracy--Widom GUE and the Gumbel
  distributions. We also define and prove convergence for a stationary
  extension of our model. Finally, we compute the correlation
  functions for a variant of the periodic Schur process involving
  strict partitions, Schur's $P$ and $Q$ functions, and neutral
  fermions.
\end{abstract}

\section{Introduction} \label{sec:intro}

The Schur process, introduced by Okounkov and Reshetikhin~\cite{or}
but also appearing more or less implicitly in the works of
Johansson~\cite{Joh02,Joh03,joh}, is in many aspects a discrete
analogue of a random matrix model such as Dyson's Brownian motion. It
is therefore not surprizing that it can be analyzed by the same
techniques and admits scaling limits in the same universality classes
(e.g.\ sine processes in the bulk, Airy processes at the edge). See
for instance~\cite{Oko02,Joh05} and references therein.

\begin{figure}[t]
  \centering
  \includegraphics[width=.5\textwidth]{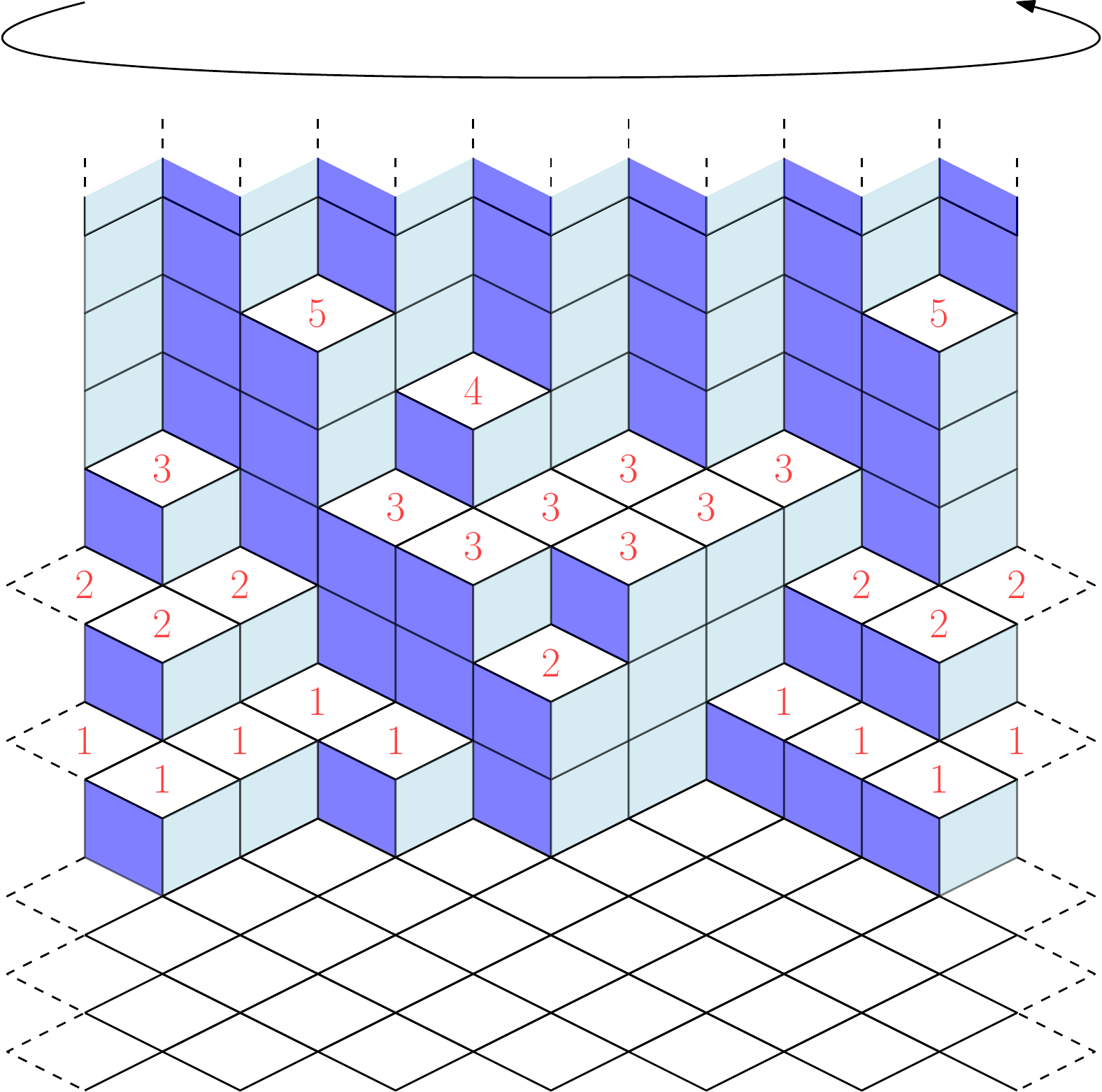}
  \caption{A cylindric partition represented as a lozenge tiling on
    the cylinder (the left/right sides are identified). The nonzero
    entries of the cylindric partition (shown in red) correspond to
    the displacements of the horizontal lozenges with respect to the
    ``ground state''. The corresponding sequence of integers
    partitions of the form~\eqref{eq:seqform} is obtained by reading
    the entries along vertical lines, namely we have
    $\mu^{(0)}=(2,1)$, $\lambda^{(1)}=(3,2,1)$, $\mu^{(1)}=(2,1)$,
    $\lambda^{(2)}=(5,1)$, etc (here $N=6$). Note the constraint that
    two consecutive partitions $\lambda,\mu$ must be
    interlaced---i.e., satisfy
    $\lambda_1 \geq \mu_1 \geq \lambda_2 \geq \mu_2 \geq \cdots$. In
    the measure~\eqref{eq:probdef} this constraint is implemented by
    taking each $\rho_k^{\pm}$ to be a specialization in a single
    variable---see \eqref{eq:singlespec}. Here the profile---in the
    sense of~\cite{B2007cyl}---is the alternating sequence
    $(1,-1,1,-1,\ldots)$.}
  \label{fig:cylpp2}
\end{figure}

In this paper we revisit the periodic Schur process introduced by
Borodin~\cite{B2007cyl}. One of his motivations was that it allows to
study random cylindric partitions in the same way as the original
Schur process allows to study random plane
partitions~\cite{or,or2}. See~\cite{gk} for the definition of
cylindric partitions and Figure~\ref{fig:cylpp2} for an illustration.

The periodic Schur process is a measure on periodic sequences of
integer partitions of the form
\begin{equation}
  \label{eq:seqform}
  \mu^{(0)} \subset \lambda^{(1)} \supset \mu^{(1)} \subset \cdots
  \supset \mu^{(N-1)} \subset \lambda^{(N)} \supset \mu^{(N)} = \mu^{(0)}
\end{equation}
such that
\begin{equation}
  \label{eq:probdef}
  \Prob(\vec{\lambda},\vec{\mu})\ \propto\ u^{|\mu^{(0)}|} 
  \prod_{k=1}^N \left(s_{\lambda^{(k)}/\mu^{(k-1)}}\left(\rho_k^+\right)
    s_{\lambda^{(k)}/\mu^{(k)}}\left(\rho_k^-\right) \right).
\end{equation}
Here $u$ is a nonnegative real parameter smaller than $1$,
$s_{\lambda/\mu}$ is a skew Schur function, and the $\rho_k^\pm$ are
collections of variables or specializations---see the beginning of
Section~\ref{sec:corr} for a summary of the relevant definitions. Note
that the constraint $u<1$ arises because constant sequences would
contribute an infinite mass otherwise. For $u=0$ the measure is
concentrated on sequences such that $\mu^{(0)}$ is the empty
partition, and we recover the original Schur process~\cite{or}.

Taking $u \neq 0$ brings an extra level of complexity: as shown by
Borodin~\cite[Theorem~A]{B2007cyl}, the point process naturally
associated with $(\vec{\lambda},\vec{\mu})$ requires a nontrivial
``shift-mixing'' transformation to be determinantal, and its
correlation functions are given by an elliptic deformation of those
for the original process. Here we will rederive this result using the
free fermion formalism\footnote{Borodin's proof involves so-called
  $L$-ensembles. As stated in~\cite{B2007cyl}, his initial derivation
  was also based on the formalism of (fermionic) Fock space, but our
  approach is different (personal communication with Alexei
  Borodin).}, which was also used recently to analyze the Schur
process with free boundary conditions~\cite{bbnv}. We believe that our
derivation is more transparent. In particular the shift-mixing
transformation follows from a passage to the grand canonical ensemble
in the fermionic picture. For pedagogical purposes, we illustrate in
Section~\ref{sec:fermidirac} the key idea of our approach for the most
elementary setting $N=0$, where the measure~\eqref{eq:probdef} reduces
to a single uniform\footnote{By a slight abuse of language, the model
  for which $\Prob(\lambda)\propto u^{|\lambda|}$ will be called
  uniform as it is the ``macrocanonical'' ensemble associated with
  uniform random partitions of fixed size~\cite{Vershik1996}. Note
  that this is not the same notion as the fermionic grand canonical
  ensemble that we will consider.} random partition.

We also investigate the edge behavior of the periodic Schur process,
the bulk case having been already analyzed by Borodin. We will
concentrate on the simplest nontrivial
instance~\cite[Example~3.4]{B2007cyl} where $N=1$ and both
$\rho_1^+,\rho_1^-$ are exponential specializations. In other words,
we consider two random partitions $\lambda \supset \mu$ such that
\begin{equation}
  \label{eq:uPlanchdef}
  \Prob(\lambda,\mu) = \frac{1}{Z_{u,\gamma}} u^{|\mu|} \left(\frac{\gamma^{|\lambda|-|\mu|} \dim(\lambda/\mu)}{(|\lambda|-|\mu|)!} \right)^2
\end{equation}
where $u \in [0,1)$, $\gamma>0$,
$Z_{u,\gamma}:=e^{\frac{\gamma^2}{1-u}}/\prod_{n \geq 1}(1-u^n)$, and
$\dim(\lambda/\mu)$ denotes the number of standard Young tableaux of
skew shape $\lambda/\mu$. The marginal distribution of $\lambda$
interpolates between the uniform measure ($\gamma=0$) and the
poissonized Plancherel measure ($u=0$) on partitions, and we call it
the \emph{cylindric Plancherel measure}.

We are interested in the thermodynamic limit ($u \to 1$ or
$\gamma \to \infty$) where the partition $\lambda$ becomes large, and
consider the behavior of the first part of $\lambda$. It is well-known
that the fluctuations of $\lambda_1$ are asymptotically given by the
Gumbel distribution in the uniform case \cite{el}, and by the
Tracy--Widom GUE distribution in the Plancherel case \cite{bdj}. We
find that, in a suitably tuned thermodynamic limit, our model provides
a ``crossover'' between these two behaviors. The interpolating
distribution depends on a positive parameter $\alpha$ and was
previously encountered by Johansson in the so-called MNS random matrix
model~\cite{joh3}, and by Le Doussal \emph{et al.}\ for free fermions
in a confining trap~\cite{DLMS16,dms}---see also~\cite{lw} and the
discussion in Section~\ref{sec:conclusion}.  It is given explicitly by
a Fredholm determinant
\begin{equation}
  \label{eq:Falphadef}
  F_\alpha(s) := \det(I-M_\alpha)_{L^2(s,\infty)}, \qquad
  M_\alpha(x,y) := \int_{-\infty}^\infty \frac{e^{\alpha v}}{1+e^{\alpha v}}
  \Ai(x+v) \Ai(y+v) dv, \qquad s,x,y \in \R
\end{equation}
where $\Ai$ is the Airy function and $\alpha$ a positive parameter.
Johansson proved that $F_\alpha$ is well-defined and indeed
interpolates between the Gumbel ($\alpha \to 0$) and the Tracy--Widom
GUE distributions ($\alpha \to \infty$).  The limit
$\alpha \to \infty$ is immediate since the ``Fermi factor''
$\frac{e^{\alpha v}}{1+e^{\alpha v}}$ reduces to the indicator
function $\mathbbm{1}_{v>0}$, and $M_\alpha$ reduces to the well-known
Airy kernel. The limit $\alpha \to 0$ is more subtle as it requires a
rescaling.  The kernel $M_\alpha$ has been called the
\emph{finite-temperature Airy kernel}~\cite{dms}. Our main result may
then be stated as follows.

\begin{thm}
  \label{thm:limedge}
  Consider the cylindric Plancherel measure~\eqref{eq:uPlanchdef} and
  let $u \to 1$ and $\gamma \to \infty$ in such a way that
  $\gamma (1-u)^2 \to \alpha^{3} > 0$ (with possibly
  $\alpha=\infty$). Then we have
  \begin{equation}
    \label{eq:limedge}
    \Prob\left( \frac{\lambda_1 - 2 L_{u,\gamma}}{L_{u,\gamma}^{1/3}} \leq s \right) \to
    F_\alpha(s), \qquad \text{where }L_{u,\gamma} := \frac{\gamma}{1-u} \sim
    \left(\frac{\gamma}{\alpha}\right)^{3/2}.
  \end{equation}
\end{thm}

A more general statement holds for the joint convergence of
$\lambda_1,\ldots,\lambda_k$ for any $k\geq 1$. Let us comment on the
``physical'' intuition behind Theorem~\ref{thm:limedge}. On the one
hand, for $u \to 1$, $\lambda_1$ will have ``thermal'' fluctuations
whose order of magnitude is $(1-u)^{-1}$. On the other hand, for
$\gamma \to \infty$, the area $|\lambda|$ concentrates around the
value $L_{u,\gamma}^2$ hence $L_{u,\gamma}$ is the natural length
scale, and we expect ``quantum'' Tracy-Widom-type fluctuations of
order $L_{u,\gamma}^{1/3}$ for $\lambda_1$. The edge crossover regime
corresponds to the situation where thermal and quantum fluctuations
are of the same order of magnitude, and we expect the resulting
distribution $F_\alpha$ to be universal. Note that the edge crossover
regime is \emph{not} the same as that for the bulk, where one takes
$u \to 1$ keeping $\gamma$ fixed. As noted by Borodin, as soon as we
take $\gamma$ large then we recover in the bulk the usual discrete
sine process, with the same density profile as in the case $u=0$
corresponding to the Logan--Shepp--Vershik--Kerov limit
shape~\cite{ls,vk}. Therefore, there exists an intermediate regime
$1 \ll \gamma \ll (1-u)^{-2}$ where the bulk behaves as at zero
temperature ($u=0$) while the edge behavior is still purely governed
by thermal fluctuations. We have a precise counterpart of
Theorem~\ref{thm:limedge} in the whole ``high temperature'' regime.

\begin{thm}
  \label{thm:limedgegum}
  For the cylindric Plancherel measure~\eqref{eq:uPlanchdef}, set $u=e^{-r}$ and assume that
  $r \to 0$ and $\gamma r^{2} \to 0$ (with $\gamma$ possibly remaining
  finite). Then the rescaled fluctuations of $\lambda_1$ follow the
  Gumbel distribution. To wit:
  \begin{equation}
    \label{eq:limedgegum}
    \Prob\left( r \lambda_1 - \ln \frac{I_0(2\gamma+\gamma r)}{r} \leq s \right) \to
    e^{-e^{-s}}
  \end{equation}
  where $I_0(x)=\frac{1}{2\pi}\int_{-\pi}^\pi e^{x \cos \phi} d \phi$
  is the modified Bessel function of the first kind and order zero.
\end{thm}

Intuitively speaking, for $0 < \gamma \ll (1-u)^{-2}$, there is still
a crossover, but at the level of the deterministic leading-order
contribution to $\lambda_1$ rather than at the level of
fluctuations. The deterministic term, which we do not expect to be
universal, is consistent with both Borodin's density profile in the
regime $\gamma=O(1)$, and with Theorem~\ref{thm:limedge} in the
intermediate regime $1 \ll \gamma \ll (1-u)^{-2}$. Indeed, in the
latter case we have
$r^{-1} \ln \frac{I_0(2\gamma+\gamma
  r)}{r}=2L_{u,\gamma}-\frac{\ln(4\pi \gamma r^2)}{2r}+o(1)$, the
logarithmic correction matching that in Johansson's formula for the
limit $\alpha \to 0$ of the $F_\alpha$
distribution~\cite[Theorem~1.3]{joh3}.

\begin{figure}[t]
  \centering
  \includegraphics{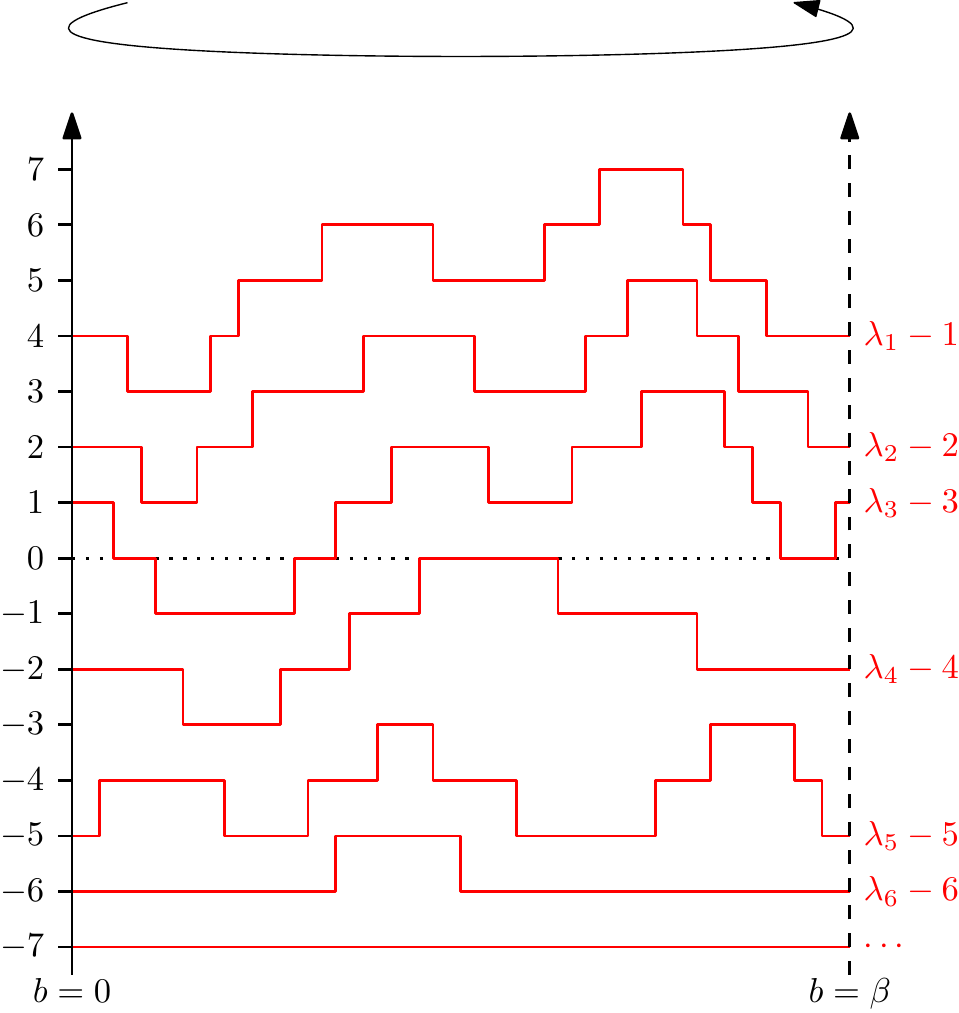}
  \caption{A possible sample of the cylindric Plancherel
    process. Following the notations of
    Section~\ref{sec:cylPlanchproc}, the horizontal axis represents
    the (imaginary) continuous time $b$, $\beta$ is the period, and
    the $i$-th path represents the time-evolution of
    $\lambda_i \in \{0,1,2,\ldots\}$ which performs unit jumps up or
    down at random times (we rather represent $\lambda_i-i$ so that
    the paths are nonintersecting). For $i$ large enough (here
    $i\geq 7$), $\lambda_i$ remains equal to $0$ at all time. For any
    fixed time $b$, the law of the partition $\lambda(b)$ is the
    cylindric Plancherel measure.}
  \label{fig:CPP_paths}
\end{figure}

Interestingly, the cylindric Plancherel measure admits a stationary
continuous-time periodic extension, which is the periodic analogue of
the stationary process of Borodin and Olshanski~\cite{BO06}, and which
we call the \emph{cylindric Plancherel process}---see
Figure~\ref{fig:CPP_paths}. We show that its correlation kernel
converges in the edge crossover regime to the extended
finite-temperature Airy kernel~\cite{dms}. The cylindric Plancherel
process can be thought as a certain ``poissonian'' limit of a measure
on cylindric partitions, as described below.

Finally, we consider an extension of our approach to the
shifted/strict setting. Vuleti\'c~\cite{vul2} defined the so-called
shifted Schur process, which is a measure on sequences of
\emph{strict} partitions (an integer partition is called strict if all
its parts are distinct), and whose definition involves Schur's $P$ 
and $Q$ functions instead of the ordinary Schur functions. We
naturally introduce the periodic variant of this process, which we
prefer to call the \emph{periodic strict Schur process}: to a sequence
of the form~\eqref{eq:seqform} where each element is a strict
partition, we assign a weight
\begin{equation}
  \label{eq:probdef_S}
  \Prob(\vec{\lambda},\vec{\mu})\ \propto\ u^{|\mu^{(0)}|} 
  \prod_{k=1}^N \left(Q_{\lambda^{(k)}/\mu^{(k-1)}}\left(\rho_k^+\right)
    P_{\lambda^{(k)}/\mu^{(k)}}\left(\rho_k^-\right) \right)
\end{equation}
with $u$ a real parameter smaller than $1$, $P_{\lambda/\mu}$ and
$Q_{\lambda/\mu}$ the Schur $P$ and $Q$ functions, and the
$\rho_k^\pm$ strict specializations---see the beginning of
Section~\ref{sec:corr_S} for definitions. Vuleti\'c's definition is
recovered in the case $u=0$, where the measure is concentrated on
sequences such that $\mu^{(0)}$ is the empty partition. We compute the
correlation functions of the periodic strict Schur process using again
the free fermion formalism that now involves so-called neutral
fermions as already observed in the $u=0$ case~\cite{vul2}. The
approach is even simpler than in the non strict case, because the
shift-mixing transformation becomes trivial. We find that the periodic
strict Schur process is a pfaffian point process, whose kernel
involves elliptic functions.

\paragraph{Outline.} Our paper is organized as follows. In
Section~\ref{sec:fermidirac} we discuss the connection between uniform
random partitions and Fermi--Dirac statistics, which illustrates the
key idea of our approach in the most elementary setting. In
Section~\ref{sec:corr}, we recall the fundamental results
of~\cite{B2007cyl} on the correlation functions of the periodic Schur
process. Our new derivation via the free fermion formalism is then
given in Section~\ref{sec:fermions}. We then turn to the asymptotic
analysis of the edge behavior in Section~\ref{sec:app}, and provide
the proof of Theorem~\ref{thm:limedge} (for
Theorem~\ref{thm:limedgegum}, some technical estimates are left to
Appendix~\ref{sec:moreasymp}). The stationary cylindric Plancherel
process is treated in
Section~\ref{sec:cylPlanchproc}. Section~\ref{sec:corr_S} is devoted
to the periodic strict Schur process and the derivation of its
correlation functions via neutral fermions. Finally
Section~\ref{sec:conclusion} gathers some concluding remarks and
discussion. A paper would not be complete without appendices: in
Appendix~\ref{sec:theta} we recall the definition of the elliptic
functions that we use; since it seems that basic facts about fermions
should always be recalled in an appendix (see e.g.\
\cite{oko,or,or2,vul}), we abide by the rule in
Appendices~\ref{sec:fermrem} and~\ref{sec:fermions_S} that deal with
charged and neutral fermions respectively; finally, in
Appendix~\ref{sec:moreasymp} we perform some further analysis of the
discrete finite-temperature Bessel kernel arising in Section~\ref{sec:app},
which we use to complete the proof of Theorem~\ref{thm:limedgegum} and
also to give a short rederivation of the bulk limiting kernel.

\paragraph{Acknowledgments.}
The authors had illuminating conversations related on the subject of
this note with many people, including I. Alevy, J. Baik, P. Biane, A. Borodin,
S. Corteel, P. Di Francesco, P. Ferrari, T. Imamura, L. Hodgkinson,
K. Johansson, C. Krattenthaler, G. Lambert, P. Le Doussal,
S. Majumdar, G. Miermont, M. Mucciconi, P. Nejjar, E. Rains,
N. Reshetikhin, T. Sasamoto, G. Schehr, M. Schlosser, M. Vuleti\'c and
M. Wheeler.

This work was initiated while the authors were
at the D\'epartement de math\'ematiques et applications, \'Ecole
normale sup\'erieure, Paris, and continued during several visits D.B.\ paid to J.B.\ at the ENS de Lyon. It was finalized while the authors were visiting the Matrix Institute on the University of Melbourne campus in Creswick, Australia. We wish to thank all institutions for their hospitality. 

We acknowledge financial support from the ``Combinatoire \`a Paris''
project funded by the City of Paris (D.B.\ and J.B.), from the
Laboratoire International Franco-Qu\'eb\'ecois de Recherche en
Combinatoire (J.B.), and from the Agence Nationale de la Recherche via
the grants ANR 12-JS02-001-01 ``Cartaplus'' and ANR-14-CE25-0014
``GRAAL'' (J.B.).

\section{Warm-up: integer partitions and Fermi--Dirac statistics}
\label{sec:fermidirac}

\begin{figure}
  \centering
  \includegraphics[width=.5\textwidth]{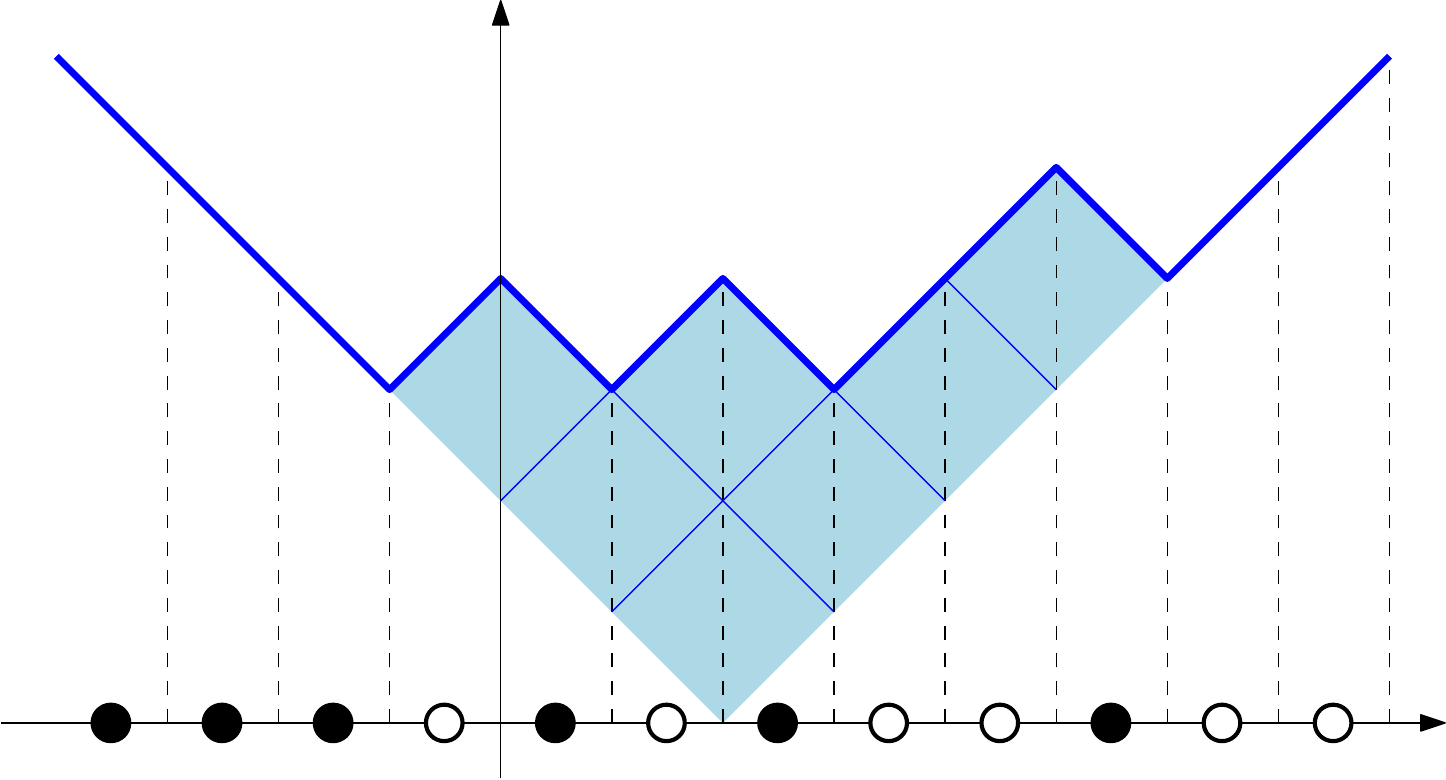}
  \caption{Illustration of the correspondence between a partition and
    a Maya diagram. The Young diagram of the partition $(4,2,1)$
    appears in blue. The corresponding Maya diagram is represented as
    a sequence of ``particles'' ($\bullet$) and ``holes''
    ($\circ$). There must be as many particles on the right of the
    bottom corner of the diagram as holes on its left. We may lift
    this constraint on the Maya diagram by moving the origin
    (displayed at the intersection between the axes) to an arbitrary
    position on the horizontal axis.  By cutting the Maya diagram at
    the origin, we obtain two strict partitions (parts correspond to
    particles on the right, holes on the left), whose number of parts
    differ by the \emph{charge} (the abscissa of the bottom corner of
    the diagram, here equal to $2$).  }
  \label{fig:russian3}
\end{figure}

An (integer) \emph{partition} $\lambda$ is a nonincreasing sequence of
integers $\lambda_1 \geq \lambda_2 \geq \cdots$ which vanishes
eventually. The strictly positive $\lambda_i$'s are called the
\emph{parts}, and their number is the \emph{length} of the partition,
denoted by $\ell(\lambda)$. The \emph{size} of $\lambda$ is
$|\lambda|:=\sum_{i \geq 1} \lambda_i$. A partition is called
\emph{strict} if all its parts are distinct. We denote by $\Par$
(respectively $\SPar$) the set of all (respectively strict) partitions. In
physics, arbitrary partitions describe ``bosons'', and strict
partitions ``fermions''~\cite{Vershik1996}.

There are several classical bijections between various classes of
partitions, strict or not. For instance, $\Par$ is in bijection with
pairs of strict partitions of the same length. Such pairs may
conveniently be seen as Maya diagrams (or excited ``Dirac seas''), and there
is more generally a bijection between $\Par \times \Z$ (\emph{charged
  partitions}) and $\SPar \times \SPar$---see
Figure~\ref{fig:russian3}. This bijection yields a beautiful proof of
the Jacobi triple product identity, see e.g.~\cite{cor} and references
therein. We now recall this proof as it is at the core of our
approach.

A \emph{Maya diagram} is a bi-infinite binary sequence
$\underline{n}=(n_i)_{i \in \Z'} \in \{0,1\}^{\Z'}$ such that $n_i=0$
for $i \gg 0$ and $n_i=1$ for $i \ll 0$. Here $\Z' := \Z +1/2$ denotes
the set of half-integers. This latter choice is purely conventional and makes
some formulas more symmetric. On Figure~\ref{fig:russian3}, $0$'s and
$1$'s are represented as holes and particles respectively. To a Maya
diagram $\underline{n}$ we associate its \emph{charge}
$C(\underline{n})$ and its \emph{energy} $H(\underline{n})$ defined by
\begin{equation}
  \label{eq:CHdef}
  C(\underline{n}) := \sum_{i>0} n_i + \sum_{i<0} (n_i-1), \qquad
  H(\underline{n}) := \sum_{i>0} i n_i + \sum_{i<0} i (n_i-1).
\end{equation}
Clearly, the energy is nonnegative and vanishes if and only if
$\underline{n}=\{\mathbbm{1}_{i<0}\}$ (\emph{vacuum}).

Given two parameters $t,u$, we associate to each Maya diagram a weight
$t^{C(\underline{n})} u^{H(\underline{n})}$. The total weight, the sum
over all Maya diagrams, clearly factorizes as an infinite product over
$i \in \Z'$. On the other hand, if we consider the charged partition
$(\lambda,c)$ associated with $\underline{n}$, then it is not
difficult to check that
\begin{equation}
  C(\underline{n}) = c, \qquad H(\underline{n}) = |\lambda| + \frac{c^2}{2}.
\end{equation}
Therefore, the weight can be rewritten as $t^c u^{|\lambda|+c^2/2}$.
By summing over all pairs $(\lambda,c)$, we end up with the identity
\begin{equation}
  \label{eq:pfdir}
  \prod_{k=0}^\infty (1 + t u^{k+1/2})(1 + t^{-1} u^{k+1/2}) =
  \left( \sum_{\lambda \in \Par} u^{|\lambda|} \right)
  \left( \sum_{c \in \Z} t^c u^{c^2/2} \right) =
  \frac{1}{(u;u)_\infty} \theta_3(t; u)
\end{equation}
which is equivalent to the Jacobi triple product identity---see
Appendix~\ref{sec:theta} for reminders and notations.

Let us now reflect on the probabilistic meaning of this
construction. Viewing $t^{C(\underline{n})} u^{H(\underline{n})}$ as a
Boltzmann weight for $\underline{n}$ and normalizing by the partition
function \eqref{eq:pfdir}, we get a well-defined probability
distribution for $t \in (0,\infty)$ and $u \in [0,1)$. In the Maya
diagram picture we learn that the $n_i$ are independent and that
\begin{equation}
  \label{eq:fermidist}
  \Prob(n_i=1) = \frac{t u^i}{1+ t u^i} = \frac{1}{1+ t^{-1} u^{-i}}.
\end{equation}
This is nothing but the \emph{Fermi--Dirac distribution} for a system
of noninteracting fermions with equally spaced energy levels (write
$t^{-1} u^{-i}=e^{\beta(\epsilon i - \mu)}$ to recognize more physical
variables). In the charged partition picture, we learn that $\lambda$
and $c$ are independent, with $\lambda$ a uniform random partition and
$c$ distributed as
\begin{equation}
  \label{eq:cdist}
  \Prob(c)=\frac{t^c u^{c^2/2}}{\theta_3(t; u)}.
\end{equation}
This fact was observed by Borodin \cite[Corollary~2.6]{B2007cyl} and
we suspect that it might appear elsewhere in the literature under
various forms.

If we condition on $c=0$, then the $n_i$'s are no longer
independent. Therefore, passing to the ``grand canonical ensemble''---as 
we do by letting $c$ fluctuate---makes it much easier to study
random partitions, at least for the observables we can capture in the
Maya diagram picture. We argue that this is the fundamental reason why
the shift-mixed periodic Schur process (to be defined below) is
determinantal: \emph{fermions are only ``free'' in the grand canonical
ensemble}.

Before proceeding, let us briefly discuss the thermodynamic limit
$u \to 1^-$. From \eqref{eq:fermidist} we see that, if we rescale
$i=x/|\ln u|$, then the Maya diagram has a limiting density profile
$1/(1+t^{-1}e^x)$. It is not difficult to check that this is
consistent with Vershik's limit shape for uniform random
partitions~\cite{Vershik1996}. Note that the parameter $t$ only
induces a shift of the profile. This is because the charge $c$
concentrates around its mean value of order $1/|\ln u|$.

\begin{lem}
  \label{lem:Cestim}
  Let $c$ be an integer-valued random variable distributed as in
  \eqref{eq:cdist}. Then, as $u \to 1^-$, the random variable
  $\tilde{c}:=(c+\ln t/\ln u)\sqrt{|\ln u|}$ converges weakly to the
  standard normal distribution.
\end{lem}

\begin{proof}
  Observe that, for any interval $I$, $\Prob(\tilde{c} \in I)$ is a
  Riemann sum converging to
  $\frac{1}{\sqrt{2\pi}} \int_I e^{-x^2/2}dx$.
\end{proof}

This shows in some sense the equivalence of the canonical and grand
canonical ensembles for our fermions. Note that the equivalence with
the microcanonical ensemble (partitions of fixed size) is another
matter~\cite{Vershik1996}, and that the equivalence between the
canonical and grand canonical ensembles in the MNS model was
established in~\cite{lw}. We now turn to the edge behavior of the Maya
diagram---see \cite[Section~3.4]{oko} for the bulk behavior.  It is
elementary to check that
\begin{equation}
  \Prob(M < k) = \frac{1}{(-t u^k;u)_\infty}, \qquad M := \max\{i: n_i=1\}.
\end{equation}
Recognizing a $q$-exponential function, and using the asymptotics
$((1-u)z;u)_\infty \to e^{-z}$ \cite[p.11]{gr}, we easily deduce
that
\begin{equation}
  \lim_{u \to 1^-} \Prob\left(M < -\frac{\ln(1-u)}{\ln u} + \frac{\xi}{|\ln u|} \right) = e^{-t e^{-\xi}}.
\end{equation}
In other words, the fluctuations of $M$ are given by the Gumbel
distribution. By Lemma~\ref{lem:Cestim}, we deduce that $\lambda_1$
has asymptotically the same distribution (with $t=1$), which is
consistent with the result of Erd\H{o}s and Lehner~\cite{el}. Note
that we could have alternatively started from the expression
$\Prob(\lambda_1 < k) = (u^k;u)_\infty$ to arrive at the same result.

\section{Correlations functions of the periodic Schur process}
\label{sec:corr}

In this section we recall the fundamental results of~\cite{B2007cyl},
for which we give a new proof in the next section.
Here we follow the notational conventions of
\cite{bbnv} regarding Schur functions, which we briefly recall now.
Given a sequence of numbers
$(h_n(\rho))_{n \in \Z}$ such that $h_0(\rho)=1$ and $h_n(\rho)=0$ for
$n<0$, and two partitions $\lambda,\mu$, the \emph{skew Schur function
  of shape $\lambda/\mu$ specialized at $\rho$} is given by
\begin{equation}
  s_{\lambda / \mu}(\rho) := \det_{1 \leq i,j \leq \ell(\lambda)} h_{\lambda_i-i-\mu_j+j}(\rho).
\end{equation}
It vanishes unless $\lambda/\mu$ is a \emph{skew shape} (i.e.\ we have
$\lambda \supset \mu$, i.e.\ $\lambda_i \geq \mu_i$ for all $i$). The
``specialization'' $\rho$ is conveniently encoded into the generating
function
\begin{equation}
  \label{eq:Hrhodef}
  H(\rho;z) := \sum_{n \geq 0} h_n(\rho) z^n.
\end{equation}
Thoma's theorem \cite{tho, ais} states that $s_{\lambda / \mu}(\rho)$
is nonnegative for all $\lambda,\mu$ if and only if the generating
function $H(\rho;z)$ is of the form
\begin{equation}
  \label{eq:rhoposparam}
  H(\rho;z) = e^{\gamma z} \prod_{i \geq 1} \frac{1+\beta_i z}{1-\alpha_i z}
\end{equation}
with $\gamma,\alpha_1,\beta_1,\alpha_2,\beta_2,\ldots$ a summable
collection of nonnegative parameters (the specialization $\rho$ is
then called \emph{nonnegative}). In particular, when only $\gamma$ is
nonzero, we obtain the \emph{exponential specialization} denoted
$\ex_\gamma$, for which
\begin{equation}
  \label{eq:expspec}
  s_{\lambda/\mu}(\ex_\gamma) = \frac{\gamma^{|\lambda|-|\mu|} \dim(\lambda/\mu)}{(|\lambda|-|\mu|)!}
\end{equation}
with $\dim(\lambda/\mu)$ the number of standard Young tableaux of
shape $\lambda/\mu$. When only $\alpha_1=q$ is nonzero, we obtain
the \emph{specialization in a single variable $q$} for which
\begin{equation}
  \label{eq:singlespec}
  s_{\lambda/\mu}(q)=
  \begin{cases}
    q^{|\lambda|-|\mu|} & \text{if $\lambda_1 \geq \mu_1 \geq \lambda_2 \geq \mu_2 \geq \cdots$,}\\
    0 & \text{otherwise}.
  \end{cases}
\end{equation}

For simplicity, we will assume in the following that all
specializations are nonnegative and such that their generating
functions are analytic and nonzero in a disk of radius $R>1$ (this is
clearly the case for the exponential specialization). In that case,
the measure \eqref{eq:probdef} can be normalized into a probability
distribution which is the \emph{periodic Schur process}---see
\cite{B2007cyl} or the discussion below for more on the normalization.

Following \cite{or}, to any sequence of partitions
$\vec{\lambda}=(\lambda^{(1)},\ldots,\lambda^{(N)})$, we associate
the \emph{point configuration}
\begin{equation}
  \label{eq:pointconfig}
  \mathfrak{S}(\vec{\lambda}) := \left\{
    \left(i,\lambda^{(i)}_j-j+\frac{1}{2}\right),\ 1 \leq i \leq N , j \geq 1
  \right\} \subset \{1,\ldots,N\} \times \Z'.
\end{equation}
Observe that this definition is closely related to the Maya diagrams
encountered in Section~\ref{sec:fermidirac}: the
$\lambda^{(i)}_j-j+\frac{1}{2}$, $j \geq 1$, are precisely the
positions of particles ($1$'s) in the Maya diagram associated with the
partition $\lambda^{(i)}$ when the charge is $0$. See again
Figure~\ref{fig:russian3}.

Our aim is to study the point configuration when $\vec{\lambda}$ is
the first marginal of the periodic Schur process. There is no loss of
generality in discarding $\vec{\mu}$, and doing so simplifies the
ensuing expressions. $\mathfrak{S}(\vec{\lambda})$ is a discrete point
process, and as such its properties are characterized by the data of
the \emph{correlation functions}
\begin{equation}
  \label{eq:rhoUdef}
  \varrho(U) := \Prob\left(U \subset \mathfrak{S}(\vec{\lambda})\right)
\end{equation}
for any finite set $U \subset \{1,\ldots,N\} \times \Z'$.

Borodin observed that $\rho(U)$ is quite complicated in general
(though it can be written in the form of a multiple contour integral---see 
Proposition~\ref{prop:ncont} below), but we may transform it into
a determinant at the price of ``shift-mixing'' the process. More
precisely, we take an integer-valued random variable $c$, independent
of $\vec{\lambda}$ and distributed according to \eqref{eq:cdist} (with
$t$ an extra arbitrary parameter), and we define the \emph{shift-mixed
  correlation function} as
\begin{equation}
  \label{eq:trhoUdef}
  \tilde{\varrho}(U) := \Prob\left(U \subset \mathfrak{S}(\vec{\lambda}) + (0,c) \right).
\end{equation}
In other words, we are just moving $\mathfrak{S}(\vec{\lambda})$ by a
random vertical shift $c$. The shift-mixed process is \emph{determinantal} in
the sense that there exists a \emph{correlation kernel}
$K: (\{1,\ldots,N\} \times \Z') \times (\{1,\ldots,N\} \times \Z') \to \C$ such that
\begin{equation}
  \tilde{\varrho}(U) = \det_{1 \leq i,j \leq n} K(u_i;u_j), \qquad U=\{u_1,\ldots,u_n\}.
\end{equation}
See for instance~\cite{Joh05} for general background on determinantal
processes. Let us now state Borodin's main result in our current
notations.

\begin{thm}[{{\cite[Theorem~A]{B2007cyl}}}] \label{thm:main_thm}
  The shift-mixed periodic Schur process is determinantal with correlation kernel given by
  \begin{equation}
    \label{eq:Kint}
    \begin{split}
      K(i,k; i',k') = \frac{1}{(2 \pi \im)^2} \oint_{|z|=r} \frac{d z}{z^{k+1}} \oint_{|w|=r'} \frac{d w}{w^{-k'+1}} \cdot \frac{F(i,z)}{F(i', w)} \cdot \kappa(z, w)
    \end{split}
  \end{equation}
  where $i,i' \in \{1,\ldots,N\}$, $k,k' \in \Z'$, the radii $r,r'$ satisfy
  \begin{equation}
    R^{-1} \leq r,r' \leq R, \qquad r/r' \in
    \begin{cases}
      (u,1) & \text{if $i > i'$,} \\
      (1,u^{-1}) & \text{if $i \leq i'$,}
    \end{cases}
  \end{equation}
  $F$ carries the dependence on the specializations $\rho_k^\pm$ as
  \begin{equation}
     \label{eq:Fthm}
     F(i, z) := \frac{\prod_{1 \leq \ell \leq i} H(\rho_{\ell}^+;z)}{\prod_{i \leq \ell \leq N} H(\rho_{\ell}^-;z^{-1})} \prod_{n \geq 1} \prod_{1 \leq \ell \leq N} \frac{H(u^n \rho_{\ell}^+;z)}{H(u^n \rho_{\ell}^-;z^{-1})}
   \end{equation}
   and
   \begin{equation}
     \label{eq:kapform}
     \kappa(z,w) :=
     \begin{cases}
       \sum_{m \in \Z'} \frac{(z/w)^m}{1+(t u^m)^{-1}} & \text{for $|z/w|\in (1,u^{-1})$,} \\
       -\sum_{m \in \Z'} \frac{(z/w)^m}{1+t u^{m}} & \text{for $|z/w| \in (u,1)$.} \\
     \end{cases}
   \end{equation}
\end{thm}

Several remarks are now in order. On the one hand, we recover the
correlation kernel of the original Schur process~\cite[Theorem~1]{or}
by taking $u=0$, in which case $\kappa(z,w)$ simplifies into
$\sqrt{zw}/(z-w)$. Note that $\Prob(c=0)=1$ so shift-mixing can be
disregarded in that case. On the other hand, if we consider only
trivial specializations (so that $F(i,z)=1$), then $K(i,k;i',k')$
vanishes for $k\neq k'$, and the shift-mixed process is a discrete
Poisson process. This is nothing but the phenomenon discussed in
Section~\ref{sec:fermidirac}.

As remarked by Borodin~\cite[Remark~2.4]{B2007cyl}, the two cases in
\eqref{eq:kapform} correspond to the expansions in two different
annuli of the same meromorphic function, namely
\begin{equation}
  \label{eq:kapprod}
  \kappa(z,w) =\sqrt{\frac{w}{z}} \cdot \frac{(u;u)^2_{\infty}}{\theta_u(\frac{w}{z})} \cdot \frac{\theta_3 (\frac{tz}{w}; u)}{\theta_3(t;u)} = \im \frac{\eta(u)^3 \theta_3(\frac{tz}{w}; u)}{
    \theta_1(\frac{w}{z}; u) \theta_3(t; u)}
\end{equation}
with $(u; u)_{\infty}$ the infinite Pochhammer symbol, $\eta$ the
Dedekind eta function, $\theta_u$ the ``multiplicative'' theta
function, and $\theta_1,\theta_3$ the Jacobi theta functions---see
Appendix~\ref{sec:theta} for conventions. The equality between the sum
expression~\eqref{eq:kapform} and the ``theta''
form~\eqref{eq:kapprod} is a particular case of Ramanujan's
${}_1\Psi_1$ summation formula, and may be rederived directly using
the boson--fermion correspondence as explained in
Section~\ref{sec:fermions}. The first theta form of $\kappa(z, w)$
arises naturally from the computations; in the second form we
introduce the Dedekind eta function to make the intriguing modular
properties more apparent. In fact, $\kappa(z,w)$ may be interpreted as
a propagator $\langle \psi(z) \psi^*(w) \rangle$ for the conformal
field theory of charged fermions on a torus, in which modular
invariance plays a crucial
role~\cite[Chapter~10]{CFTbook}.

Finally, Borodin observed that using an elliptic version of
the Cauchy determinant due to Frobenius---see
Remark~\ref{rem:frobdet} below---, we may obtain an explicit contour
integral representation for $\varrho(U)$:

\begin{prop}[{{\cite[Corollary~2.8]{B2007cyl}}}]
  \label{prop:ncont}
  The general $n$-point correlation function for the periodic Schur process is
  given by
  \begin{equation}
    \varrho(U) = \frac{\eta(u)^{3n}}{(2\im \pi)^{2n}} \oint \cdots \oint
    \prod_{\ell=1}^n \left( \frac{d z_\ell}{z_\ell^{k_\ell+1}} \cdot \frac{d w_\ell}{w_\ell^{-k_\ell+1}} \cdot \frac{F(i_\ell,z_\ell)}{F(i_\ell, w_\ell)} \right) \cdot
    \frac{\prod_{1 \leq \ell < m \leq n} \theta_1(z_\ell/z_m;t) \theta_1(w_\ell/w_m;t)}{
      \prod_{1 \leq \ell,m \leq n} \theta_1(z_\ell/w_m;t)}
  \end{equation}
  where $U=\{(i_1,k_1),(i_2,k_2),\ldots,(i_n,k_n)\}$ and, upon
  ordering $i_1 \leq i_2 \leq \cdots \leq i_n$, the integration is taken over nested
  circles $\min(u^{1/2},R)>|z_1|>|w_1|>|z_2|>|w_2|>\cdots>|z_n|>|w_n|>\max(u^{-1/2},R^{-1})$. The
  $n$-point correlation $\tilde{\rho}(U)$ for the shift-mixed process
  is given by simply multiplying the integrand above by
  $\theta_3(t\frac{z_1 z_2 \cdots z_n}{w_1 w_2 \cdots
    w_n};u)/\theta_3(t;u)$.
\end{prop}

\section{Derivation via free fermions}
\label{sec:fermions}

This section is devoted to a proof of Theorem~\ref{thm:main_thm} and
Proposition~\ref{prop:ncont} via the machinery of free fermions. Again
we make use of the same conventions and notations as in~\cite{bbnv}---for 
convenience the relevant definitions and facts are recalled in
Appendix~\ref{sec:fermrem}.

Our starting point is the observation that the \emph{partition
  function}---i.e., the sum of the unnormalized weights
\eqref{eq:probdef}---of the periodic Schur process can be
represented as
\begin{equation}
  \label{eq:Zgam}
    Z = \tr \left( \Pi_0 u^H \Gamma_+(\rho_1^+) \Gamma_-(\rho_1^-) \cdots
    \Gamma_+(\rho_N^+) \Gamma_-(\rho_N^-) \right)
\end{equation}
where $\tr$ stands for the trace over the fermionic Fock space and
$\Pi_0$ denotes the orthogonal projector over the subspace of charge 0
(as we shall sum over ordinary \emph{uncharged} partitions). Similarly, for
$U = \{(i_1, k_1),\dots,(i_n, k_n)\} \subset \{1,\ldots,N\} \times
\Z'$, if the abscissas are ordered as $i_1 \leq \dots \leq i_n$, then
the correlation function $\varrho(U)$ as defined in \eqref{eq:rhoUdef}
is equal to $Z_U/Z$, where
\begin{equation}
  \label{eq:ZUdef}
  Z_U := \tr \left( \Pi_0 u^H \Gapl(\rho^+_1) \cdots \Gapl(\rho^+_{i_1}) \psi_{k_1} \psi^*_{k_1} \Gami(\rho^-_{i_1}) \cdots \Gapl(\rho^+_{i_n}) \psi_{k_n} \psi^*_{k_n} \Gami(\rho^-_{i_n}) \cdots \Gami(\rho^-_{N}) \right)
\end{equation}
is the sum of the unnormalized weights of sequences such that
$U \subset \mathfrak{S}(\vec{\lambda})$.

Let us denote by $\tilde{Z}$ and $\tilde{Z}_U$ the quantities obtained
by replacing respectively in \eqref{eq:Zgam} and \eqref{eq:ZUdef} the
projector $\Pi_0$ by the operator $t^C$ . Then, it is not difficult to
check that $\tilde{Z}_U/\tilde{Z}$ is precisely the shift-mixed
correlation function $\tilde{\rho}(U)$ as defined in
\eqref{eq:trhoUdef}.

\paragraph{Eliminating the $\Gamma$-operators.} We now rewrite the
quantities $Z,Z_U,\tilde{Z},\tilde{Z}_U$ following the strategy of
\cite{or}, which was already adapted to the periodic setting in
\cite{bcc} (in the case of $Z$). We call this method
\emph{$\Gamma$-elimination}: in a nutshell if we consider the trace of
a product of operators involving some $\Gamma_\pm$'s and other
operators which they quasi-commute with (such as
$u^H,\psi(z),\psi^*(w)$), then the quasi-commutation relations can be
exploited to remove the $\Gamma$'s from the product, up to a
multiplication by the corresponding scalar factors.

Let us illustrate this method on the easiest case of $Z$.  The basic
idea is to move each of the $\Gamma_+(\rho_i^+)$ to the right, using
the quasi-commutation relations~\eqref{eq:gamcomm} with the other
operators and the cyclic property of the trace (recall also that
$\Gamma$-operators commute with $C$ hence $\Pi_0$). For instance,
performing this operation on each $\Gamma_+(\rho_i^+)$ in
\eqref{eq:Zgam} until it goes back into the same place, we get
\begin{equation}
  Z  = \prod_{1 \leq i \leq j \leq N} H(\rho_i^+;\rho_j^-) \prod_{1 \leq j < i \leq N} H(u \rho_i^+;\rho_j^-) \times \tr \left( \Pi_0 u^H \Gamma_+(u \rho_1^+) \Gamma_-(\rho_1^-) \cdots
    \Gamma_+(u \rho_N^+) \Gamma_-(\rho_N^-) \right).
\end{equation}
(note that the specializations $\rho_i^+$ have all been multiplied by
$u$ in the trace on the right). By iterating this procedure $m$ times,
and noting that the process converges for $m \to \infty$ because
$\Gamma_+(u^m \rho) \to 1$ for $m \to \infty$, we may rederive
Borodin's expression for the partition function
\cite[Proposition~1.1]{B2007cyl}---see \cite[proof of
Theorem~12]{bcc} for more details.

For the correlation functions, it is useful to first rewrite
\begin{multline}
  \label{eq:ZUrew}
  Z_U = \left[ z_1^{k_1} w_1^{-k_1} \cdots z_n^{k_n} w_n^{-k_n}\right]\tr \left( \Pi_0 u^H \Gapl(\rho^+_1) \cdots \Gapl(\rho^+_{i_1}) \psi(z_1) \psi^*(w_1) \Gami(\rho^-_{i_1}) \cdots \right. \\ \left. \cdots \Gapl(\rho^+_{i_n}) \psi(z_n) \psi^*(w_n) \Gami(\rho^-_{i_n}) \cdots \Gami(\rho^-_{N}) \right)
\end{multline}
where $[ z_1^{k_1} w_1^{-k_1} \cdots z_n^{k_n} w_n^{-k_n}]$
denotes coefficient extraction in the multivariate Laurent series on
the right. We may now apply the same strategy to the trace in
\eqref{eq:ZUrew}, and we will pick extra factors of the form
$H(u^m \rho_i^+;z_\ell)/H(u^m \rho_i^+;w_\ell)$ from the
quasi-commutations between $\Gamma_+$'s with
$\psi(z_\ell) \psi^*(w_\ell)$'s. Note that the factors
$H(u^m \rho_i^+;\rho_j^-)$ coming from quasi-commutations between
$\Gamma_+$'s and $\Gamma_-$'s will eventually get cancelled when we
divide by $Z$ to get $\rho(U)$. Furthermore, after we have
``eliminated'' the $\Gamma_+$'s, we also need to get rid of the
$\Gamma_-$'s by moving them similarly but to the left. This will
produce factors of the form
$H(u^m \rho_i^-;w_\ell^{-1})/H(u^m \rho_i^-;z_\ell^{-1})$. We end up
with the expression
\begin{equation}
  \label{eq:rhoUF}
  \rho(U) =  \left[ z_1^{k_1} w_1^{-k_1} \cdots z_n^{k_n} w_n^{-k_n}\right]
  \prod_{\ell=1}^n \frac{F(i_\ell,z_\ell)}{F(i_\ell,w_\ell)} \cdot
  \evuc{\psi(z_1) \psi^*(w_1)  \cdots  \psi(z_n) \psi^*(w_n)}
\end{equation}
where $\evuc{O} := \tr (\Pi_0 u^H O)/\tr( \Pi_0 u^H)$ denotes the
``canonical'' expectation value of the operator $O$, and where $F$ is
as in Theorem~\ref{thm:main_thm}: the $F(i_\ell,z_\ell)$ in the
numerator (respectively the $F(i_\ell,w_\ell)$ in the denominator) contains
all the factors arising from the quasi-commutations of $\Gamma$'s with
$\psi(z_\ell)$ (respectively $\psi^*(w_\ell)$) in \eqref{eq:ZUrew}.

All this reasoning remains valid if we consider the shift-mixed
correlation function $\tilde{\rho}(U)$, which then admits the
expression obtained by replacing in \eqref{eq:rhoUF} the canonical
expectation value $\evuc{\cdot}$ by the grand canonical one
$\evut{\cdot} := \tr (t^C u^H \cdot)/\tr( t^C u^H)$.

\paragraph{Determinantal structure of $\tilde{\rho}(U)$.} At this
stage we may employ Wick's lemma in finite temperature
(Lemma~\ref{lem:wickgen}) in its ``determinantal'' form to conclude
that the shift-mixed correlation function can be rewritten as
\begin{equation}
  \label{eq:trhodet}
  \begin{split}
  \tilde{\rho}(U) &=  \left[ z_1^{k_1} w_1^{-k_1} \cdots z_n^{k_n} w_n^{-k_n}\right]
  \prod_{\ell=1}^n \frac{F(i_\ell,z_\ell)}{F(i_\ell,w_\ell)} \cdot
  \det_{1 \leq \ell,m \leq n} \evut{\mathcal{T}(\psi(z_{i_\ell}),\psi^*(w_{i_m}))} \\
  &= \det_{1 \leq \ell,m \leq n} K(i_\ell,k_\ell;i_m,k_m), \quad
  K(i,k;i';k') := \left[ z^k w^{-k'} \right] \frac{F(i,z)}{F(i',w)} \cdot
  \begin{cases}
    \evut{\psi(z)\psi^*(w)} & i \leq i'\\
    -\evut{\psi^*(w)\psi(z)} & i > i'
  \end{cases}
  .
\end{split}
\end{equation}
This shows that the shift-mixed process is indeed determinantal. We
may rewrite $K(i,k;i';k')$ in the form~\eqref{eq:Kint} as follows.
From the discussion of Section~\ref{sec:fermidirac} and particularly
\eqref{eq:fermidist}---or alternatively from \eqref{eq:anticomcor}---we 
see that
\begin{equation}
  \label{eq:propag}
  \evut{\psi_k \psi_k^*} = \frac{t u^k}{1+t u^k}=\frac{1}{1+t^{-1} u^{-k}},
  \qquad
  \evut{\psi_k^* \psi_k} = \frac{1}{1+t u^k},
  \qquad
  k \in \Z'
\end{equation}
while the expectation value of any other product of two fermionic
operators vanishes. Passing to the generating series we get
\begin{equation}
  \label{eq:propagzw}
  \evut{\psi(z) \psi^*(w)} = \sum_{m \in \Z'} \frac{(z/w)^m}{1+(t u^m)^{-1}}, \qquad
  \evut{\psi^*(w) \psi(z)} = \sum_{m \in \Z'} \frac{(z/w)^m}{1+t u^{m}}.
\end{equation}
Note that, at an analytic level, the first sum converges for
$|uz|<|w|<|z|$ and the second for $|uw|<|z|<|w|$, and in fact we have
\begin{equation}
  \label{eq:kferm}
  \kappa(z,w) =
  \begin{cases}
    \evut{\psi(z) \psi^*(w)} & \text{for $|z/w|\in (1,u^{-1})$,} \\
    -\evut{\psi^*(w) \psi(z)} & \text{for $|z/w| \in (u,1)$} \\
  \end{cases}
\end{equation}
(by the canonical anticommutation relations, it is a general fact that
fermionic propagators $\langle\psi(z) \psi^*(w) \rangle$ and
$-\langle\psi^*(w) \psi(z)\rangle$ should be two Laurent series
expansions of a same meromorphic function with a pole of residue $1$
at $z=w$).  We conclude the proof of Theorem~\ref{thm:main_thm} by
representing the coefficient extraction $[z^k w^{-k'}]$ in
\eqref{eq:trhodet} as a double contour integral. Note that, by our
analyticity assumptions, $F(i,\cdot)$ is analytic and nonzero in the
annulus $R^{-1}<|\cdot|<R$.

\paragraph{Theta form for fermionic expectations.} Interestingly, it
is possible to derive the ``theta'' form~\eqref{eq:kapprod} of
$\kappa(z,w)$ from the boson---fermion
correspondence~\eqref{eq:boson_fermion}. Assuming $|z|>|w|$ we may
write
\begin{equation}
  \label{eq:kapbf}
  \begin{split}
    \tr \left( t^C u^H \psi(z) \psi^*(w) \right)  &= \sqrt{\frac{w}{z}} \cdot
    \tr \left( \left(\frac{tz}{w}\right)^C u^H  \Gami(z) \Gatpl \left(-z^{-1} \right) \Gatmi(-w) \Gapl \left( w^{-1} \right)\right) \\
    &= \sqrt{\frac{w}{z}} \cdot \frac{(u;u)^2_{\infty}}{\theta_u(\frac{w}{z})} \cdot \tr \left( \left(\frac{tz}{w}\right)^C u^H \right)
  = \sqrt{\frac{w}{z}} \cdot \frac{(u;u)^2_{\infty}}{\theta_u(\frac{w}{z})} \cdot \frac{\theta_3 (\frac{tz}{w};u)}{(u;u)_{\infty}}
  \end{split}
\end{equation}
where we pass to the second line by $\Gamma$-elimination, and to the
final form by using \eqref{eq:pfdir} with $t \to tz/w$. Upon dividing
by the normalization
$\tr \left( t^C u^H \right) = \theta_{3}(t;u)/(u;u)_{\infty}$ we
arrive at~\eqref{eq:kapprod}.

Remarkably, the same trick allows to evaluate the canonical and grand
canonical expectation value of the product of any number of fermionic
generating series. We find
\begin{equation}
  \label{eq:Phiexp}
  \evuc{\psi(z_1) \psi^*(w_1)  \cdots  \psi(z_n) \psi^*(w_n)} =
  \sqrt{\frac{w_1\cdots w_n}{z_1 \cdots z_n}}
  \frac{(u;u)_\infty^{2n}  \prod_{1 \leq i < j \leq n} \theta_{u}(z_j/z_i, w_j/w_i)}  { \prod_{1 \leq i \leq j \leq n} 
    \theta_u (w_j/z_i)  \prod_{n \geq i > j \geq 1} \theta_u (z_i/w_j)}.
\end{equation}
and
\begin{equation}
  \label{eq:Phitexp}
  \evut{\psi(z_1) \psi^*(w_1)  \cdots  \psi(z_n) \psi^*(w_n)} = \theta_3 \left( t \frac{z_1 \cdots z_n}{w_1 \cdots w_n}; u\right) \frac{\evuc{\psi(z_1) \psi^*(w_1)  \cdots  \psi(z_n) \psi^*(w_n)}}{\theta_3(t;u)}.
\end{equation}
Note that we should have $u < |w_j/z_i| < 1$ (respectively
$1 < |w_j/z_i| < u^{-1}$) for $i \leq j$ (respectively for $i>j$) in
order for the expectation values to be well-defined.  By plugging
these expressions in \eqref{eq:rhoUF} and its shift-mixed counterpart,
and writing the coefficient extractions as contour integrals, we
obtain Proposition~\ref{prop:ncont}.

\begin{rem}
  \label{rem:frobdet}
  In fact, our computation yields a fermionic proof of the Frobenius
  elliptic determinant identity~\cite{fro}: on the one hand, using
  Wick's lemma, the correlator
  $\evut{\psi(z_1) \psi^*(w_1) \cdots \psi(z_n) \psi^*(w_n)}$ may be
  written as the determinant
  $\det_{1 \leq i,j \leq n} \kappa(z_i,w_j)$ whose entries can be put
  in the theta form~\eqref{eq:kapprod}; on the other hand we obtain a
  product of theta functions through \eqref{eq:Phitexp} and \eqref{eq:Phiexp}.
\end{rem}

\section{Edge behavior of the cylindric Plancherel measure}
\label{sec:app}

The purpose of this section is to establish Theorems~\ref{thm:limedge}
and~\ref{thm:limedgegum}. Recall that we want to study the asymptotic
distribution of $\lambda_1$ in the cylindric Plancherel measure, which
is the $\lambda$-marginal of the measure \eqref{eq:uPlanchdef}. By
\eqref{eq:expspec}, this measure is a periodic Schur process of period
$N=1$ with $\rho_1^\pm=\ex_\gamma$ the exponential specialization. We
start of course by applying Theorem~\ref{thm:main_thm}, which entails
that the (one-dimensional) shift-mixed point configuration
\begin{equation}
  \mathfrak{S}(\lambda)+c = \left\{ \lambda_j - j + \frac{1}{2} + c,\ j \geq 1 \right\} 
\end{equation}
is a determinantal point process with correlation kernel
\begin{align}
  K(a,b) &= \frac{1}{(2 \pi \im)^2} \oint_{|z|=u^{-1/4}} \frac{d z}{z^{a+1}} \oint_{|w|=u^{1/4}} \frac{d w}{w^{-b+1}} \cdot \frac{e^{L(z-z^{-1})}}{e^{L(w-w^{-1})}} \cdot \kappa(z, w) \label{eq:Kcontour} \\
         &= \sum_{\ell \in \Z'} \frac{J_{a+\ell}(2L) J_{b+\ell}(2L)}{1+t^{-1} u^{\ell}}
           \label{eq:KBessel}
\end{align}
where $a,b \in \Z'$, $L=\gamma/(1-u)$ (this is the $L_{u,\gamma}$ from
Theorem~\ref{thm:limedge}, but we drop the indices to lighten the
notations) and $J_n$ denotes the Bessel function of the first
kind. The choice of the integration radii is somewhat arbitrary but
will prove convenient for the forthcoming analysis. We pass to the
second form by using the definition \eqref{eq:kapform} of $\kappa$ and
the Laurent series expansion
$e^{L(z-z^{-1})} = \sum_{n \in \Z} J_n (2L) z^{n}$. Note that, for
$u \to 0^+$, the factor $1/(1+t^{-1} u^{\ell})$ becomes the indicator
function of $\Z'_+$ and $K(a, b)$ becomes the celebrated Bessel kernel
of Borodin--Okounkov--Olshanski~\cite{boo} and
Johansson~\cite{joh4}. As such and in analogy to Johansson's
``finite-temperature GUE kernel''~\cite{joh3}, we call the kernel
$K(a,b)$ the \emph{discrete finite-temperature Bessel kernel}.

\paragraph{Elementary properties of the discrete finite-temperature
  Bessel kernel.}

We note first that $K$ is symmetric and positive semi-definite. This
latter property is most evident from the Bessel representation: for
any complex numbers $z_1, \dots, z_n$ we have
\begin{equation}
  \sum_{i,j=1}^{n} z_i \overline{z_j} K(a_i, a_j) = \sum_{i,j=1}^{n} \sum_{\ell \in \Z'} z_i \overline{z_j} \frac{J_{a_i + \ell}(2L) J_{a_j + \ell}(2L)}{1+t^{-1} u^{\ell}} = \sum_{\ell \in \Z'} \frac{1}{1+t^{-1} u^{\ell}} \left| \sum_{i=1}^n z_i J_{a_i + \ell}(2S) \right|^2 \geq 0. 
\end{equation}
Second, for any $m \in \Z'$, $K$ is trace-class on
$\ell^2(\{ m, m+1, \dots\})$: here it is simpler to use the contour
integral representation which immediately implies, by crudely bounding
the integrand, that $K(a,b)=O(u^{-(a+b)/4})$ for $a,b \to +\infty$,
hence we have $\sum_{i \geq m} |K(i,i)|< \infty$.

It follows that the Fredholm determinant
$\det(I-K)_{\ell^2(\{ m, m+1, \dots\})}$ is well-defined, and is equal
to the probability that $\lambda_1+c<m$ in the shift-mixed cylindric
Plancherel measure. The parameter $t$ governing the
distribution~\eqref{eq:cdist} of $c$ does not play a significant role
in the asymptotic analysis and we will set it to $1$ from now on.

\paragraph{Length scales and asymptotics.}

Before entering into the precise asymptotic analysis of $K$, it is
useful to have a qualitative discussion of length scales. From the
expression $Z_{u,\gamma}=e^{\frac{\gamma^2}{1-u}}/(u;u)_\infty$ of the
partition function of the cylindric Plancherel
measure~\eqref{eq:uPlanchdef}, it is straightforward to check that
\begin{equation}
  \label{eq:CPMexp}
  \mathbb{E}|\lambda|=\frac{\gamma^2}{(1-u)^2}- u \frac{d}{du} \ln (u;u)_\infty, \qquad
  \mathrm{Var}|\lambda|=\frac{\gamma^2(1+u)}{(1-u)^3}- \left(u \frac{d}{du}\right)^2 \ln (u;u)_\infty.
\end{equation}
From the asymptotics $\ln (u;u)_\infty \sim - \frac{\pi^2}{6(1-u)}$
for $u \to 1^-$, we see that
\begin{equation}
  \label{eq:CPMconc}
  \mathbb{E}|\lambda| \sim \frac{\gamma^2+\frac{\pi^2}{6}}{(1-u)^2}=:\Lambda_{u,\gamma}, \qquad
  \mathrm{Var}|\lambda| \ll (\mathbb{E}|\lambda|)^2
\end{equation}
whenever $u \to 1^-$ or $\gamma \to \infty$. This shows that
$|\lambda|/\Lambda_{u,\gamma}$ tends to $1$ in probability. Note that
$\Lambda_{u,\gamma} \sim L^2$ for $\gamma \to \infty$, hence $L$ is
the natural length scale for the Young diagram of $\lambda$ and for
the point process $\mathfrak{S}(\lambda)$. For $u \to 1^-$ and fixed
$\gamma>0$, $L$ remains a good length scale but it is easier to work
with $r^{-1}$, where $u=e^{-r}$.

Borodin has shown the following ``bulk'' asymptotics for the one-point
function~\cite[Example~3.4]{B2007cyl}:
\begin{equation}
  \label{eq:Borbulk}
  \lim_{\substack{u=e^{-r},\ r \to 0^+\\\gamma \text{ fixed, }t=1}} K\left(\left\lfloor \frac{\tau}{r} \right\rfloor, \left\lfloor \frac{\tau}{r} \right\rfloor\right) = \frac{1}{2\pi} \int_{-\pi}^{\pi} \frac{d\phi}{1+e^{\tau-2\gamma\cos \phi}}=:\rho(\tau), \qquad \tau \in \R.
\end{equation}
In fact he obtained more generally the bulk limiting kernel---see
Proposition~\ref{prop:Borbulker}. For $\tau \to \infty$ we have
$\rho(\tau) \sim I_0(2\gamma) e^{-\tau}$ and we expect to find the
``edge'' (rightmost particle) at a position $\frac{\tau_e+O(1)}{r}$
with $\rho(\tau_e)=r$, i.e.\ $\tau_e=\ln \frac{I_0(2\gamma)}{r}$. This
provides some intuition regarding the spatial rescaling in
Theorem~\ref{thm:limedgegum}. For $\gamma \to \infty$, the limiting
density $\rho$ tends to that of the Logan--Shepp--Vershik--Kerov limit
shape, namely
\begin{equation}
  \label{eq:LSVK}
  \lim_{\gamma \to \infty} \rho(\gamma x) =
  \begin{cases}
    0 & \text{if $x \geq 2$,}\\
    \arccos(x) & \text{if $-2<x<2$,}\\
    1 & \text{if $x \leq -2$.}
  \end{cases}
\end{equation}
and we now expect to find the edge around the point
$2\frac{\gamma}{r} \sim 2L$, consistently with
Theorem~\ref{thm:limedge}. Of course, understanding fluctuations
requires more precise asymptotics, which we provide now.

To establish Theorem~\ref{thm:limedge}, we shall show that $K$
converges to the finite-temperature Airy kernel $M_\alpha$ in the edge
crossover regime, as stated in the following.

\begin{prop}
  \label{prop:Kconv}
  For $t=1$ and in the edge crossover regime
  \begin{equation}
    \label{eq:edgescal}
   \left\{
      \begin{array}[c]{l}
        L \to \infty \\ u \to 1^- \\
        L^{1/3}(1-u) \to \alpha > 0
      \end{array} \right.
    \qquad
    \left\{
      \begin{array}[c]{l}
        a=\lfloor 2L+xL^{1/3} \rfloor \\
        b=\lfloor 2L+yL^{1/3} \rfloor \\
        m=\lfloor 2L+sL^{1/3} \rfloor
     \end{array}
   \right. ,
  \end{equation}
  we have
  \begin{equation}
    \label{eq:Kconv}
    L^{1/3} K(a,b) \to M_{\alpha}(x,y)
  \end{equation}
  and
  \begin{equation}
    \label{eq:Ktrconv}
    \sum_{i \geq m} K(i,i) \to
    \int_{s}^\infty M_{\alpha}(s',s') d s'.
  \end{equation}
\end{prop}

Note that the right-hand side of \eqref{eq:Ktrconv} is finite
by~\cite[Proposition~1.1]{joh3}.
For good measure, we will give below \emph{two} proofs of
Proposition~\ref{prop:Kconv}. The first proof uses the Bessel
representation, and is an adaptation of the zero-temperature proofs
from~\cite{boo, joh4}---see Romik's book~\cite{rom} for a pedagogical
exposition. The second proof is based on the contour integral
representation. Let us now give the analogous statement, implying
Theorem~\ref{thm:limedgegum}, for the high temperature regime.

\begin{prop}
  \label{prop:Kconvgum}
  For $t=1$ and in the edge high temperature regime
  \begin{equation}
    \label{eq:edgescalgum}
   \left\{
     \begin{array}[c]{l}
       u = e^{-r} \\
       r \to 0^+\\
       L = o(r^{-3})
      \end{array} \right.
    \qquad
    \left\{
      \begin{array}[c]{l}
        a=\lfloor M +x r^{-1} \rfloor \\
        b=\lfloor M +y r^{-1} \rfloor \\
        m=\lfloor M +s r^{-1} \rfloor
     \end{array}
   \right. , \qquad M:=r^{-1} \ln \frac{I_0(2Lr)}{r},
  \end{equation}
  we have
  \begin{equation}
    \label{eq:Kconvgum}
    r^{-1} K(a,b) \to
    \begin{cases}
      e^{-x} & \text{if $x=y$,}\\
      0 & \text{otherwise,}
    \end{cases}
  \end{equation}
  and
  \begin{equation}
    \label{eq:Ktrconvgum}
    \sum_{i \geq m} K(i,i) \to e^{-s}.
  \end{equation}
\end{prop}

The proposition says essentially that, near the edge in the high
temperature regime, the shift-mixed determinantal process degenerates
to a Poisson point process with exponentially decreasing intensity.
See also the discussion in~\cite{joh3}. The proof is given in
Appendix~\ref{sec:moreasymp} and is based a slightly nonstandard
analysis of the contour integral representation for $K$, which may
have an interest on its own.

\paragraph{Proof that Proposition~\ref{prop:Kconv} implies
  Theorem~\ref{thm:limedge}.}

Observe that, to establish the convergence in distribution
\eqref{eq:limedge}, we may replace $\lambda_1$ by the maximum of
shift-mixed point configuration at $t=1$, $\lambda_1+c-1/2$. Indeed,
the difference between the two rescaled random variables is
$(c-1/2)/L^{1/3}$ which, by Lemma~\ref{lem:Cestim}, converges to zero
in probability for $t=1$. (And taking $t \neq 1$ fixed simply amounts
to a deterministic shift.)

Since the shift-mixed point process is determinantal, we are left with
the task of proving the convergence of Fredholm determinants
\begin{equation}
  \label{eq:detconv}
  \det(I-K)_{\ell^2(\{ m, m+1, \dots\})} \to \det(I-M_\alpha)_{L^2(s,\infty)} = F_\alpha(s)
\end{equation}
in the scaling regime~\eqref{eq:edgescal}. This is done by standard
arguments, which we recall for convenience. Set
$\tilde{K}(x,y) := L^{1/3} K(a,b)$ with $a,b$ as in
\eqref{eq:edgescal}: by \eqref{eq:Kconv}, $\tilde{K}(x,y)$ converges
pointwise to $M_\alpha(x,y)$, while we have
\begin{equation}
  \det(I-K)_{\ell^2(\{ m, m+1, \dots\})} = \sum_{n = 0}^\infty \frac{(-1)^n}{n!}
  \int_{(s',\infty)^n} \det_{1 \leq i,j \leq n} \tilde{K}(x_i,x_j) d x_1 \cdots d x_n .
\end{equation}
with $s'=(m-2L)/L^{1/3} \to s$. The integrand converges pointwise to
$\det_{1 \leq i,j \leq n} M_\alpha(x_i,x_j)$, and Hadamard's
inequality states that its modulus is bounded by
$\prod_{i=1}^\infty \tilde{K}(x_i,x_i)$ (recall that $K$ is positive
semi-definite), so we may conclude using dominated convergence and the
convergence of traces~\eqref{eq:Ktrconv} that
\begin{equation}
    \det(I-K)_{\ell^2(\{ m, m+1, \dots\})} \to \sum_{n = 0}^\infty \frac{(-1)^n}{n!}
  \int_{(s,\infty)^n} \det_{1 \leq i,j \leq n} M_\alpha(x_i,x_j) d x_1 \cdots d x_n  = \det(I-M_\alpha)_{L^2(s,\infty)}
\end{equation}
as desired. \qed

\begin{rem}
  An extension of this argument shows that the joint law of the $k$
  first parts of $\lambda$ converges to that of the $k$ righmost
  particles in the finite-temperature Airy process. See~\cite[Section
  4.2]{boo} for the zero-temperature case.
\end{rem}

\paragraph{Proof that Proposition~\ref{prop:Kconvgum} implies
  Theorem~\ref{thm:limedgegum}.}

The argument of the previous proof applies mutatis-mutandis. Moreover, we
are now able to explicitly evaluate the limiting Fredholm determinant:
\begin{equation}
  \det(I-K)_{\ell^2(\{ m, m+1, \dots\})} \to \sum_{n = 0}^\infty \frac{(-1)^n}{n!} e^{-ns} = e^{-e^{-s}}
\end{equation}
since
$\det_{1 \leq i,j \leq n} \tilde{K}(x_i,x_j) \to
e^{-(x_1+\cdots+x_n)}$ almost everywhere. Note that, in the
regime~\eqref{eq:edgescalgum}, we have $2Lr=2\gamma+\gamma r+o(1)$ so
the deterministic leading-order term $M$ is essentially the same as
that in Theorem~\ref{thm:limedgegum}.  \qed

\paragraph{Proof of Proposition~\ref{prop:Kconv} via the Bessel
  representation.}

  We shall use Nicholson's approximation (found in~\cite[Section 8.43]{wat} but see~\cite[Lemma 4.4]{boo} for a recent elementary proof), as stated in~\cite[Theorem 2.27]{rom}: 
  \begin{itemize}
    \item for $x \in \R$, uniformly on compact sets,
  \begin{equation}
    \label{eq:nic1}
    L^{1/3} J_{2L + x L^{1/3}}(2L) \to \Ai(x),\qquad L \to \infty;
  \end{equation}
    \item there exist constants $c_1, c_2, L_0$ such that for $L > L_0$ and $x > 0$,
    \begin{equation}
    \label{eq:nic2}
    L^{1/3} |J_{2L + x L^{1/3}}(2L)| < c_1 \exp(-c_2 x).
    \end{equation}
  \end{itemize}
  In the Bessel representation~\eqref{eq:KBessel}, we split the sum
  into three parts
  \begin{equation}
    L^{1/3} \sum_{\ell \in \Z'} = L^{1/3} \left( \sum_{\ell=-\infty}^{-T L^{1/3}-1} + \sum_{\ell=-T L^{1/3}}^{T L^{1/3}} + \sum_{\ell=T L^{1/3}+1}^{\infty} \right) =: \Sigma_1 + \Sigma_2 + \Sigma_3.
  \end{equation}
  for some large enough $T$ yet to be chosen.

  The main asymptotic contribution comes from $\Sigma_2$ which, using Nicholson's approximation~\eqref{eq:nic1}, becomes a Riemann sum for the finite-temperature Airy integral:
  \begin{equation} \label{eq:sig_1}
    \begin{split}
    \Sigma_2 &\approx L^{-1/3} \sum_{\ell=-T L^{1/3}}^{T L^{1/3}} \left(L^{1/3} J_{2L + x L^{1/3} + \ell} (2 L) \right) \left( L^{1/3} J_{2L + y L^{1/3} + \ell} (2 L)\right) \frac{e^{\alpha L^{-1/3} \ell} } {1 + e^{\alpha L^{-1/3} \ell}} \\
    &\approx L^{-1/3} \sum_{\ell=-T L^{1/3}}^{T L^{1/3}}  \Ai(x+\ell L^{-1/3}) \Ai(y+\ell L^{-1/3}) \frac{e^{\alpha L^{-1/3} \ell} } {1+e^{\alpha L^{-1/3} \ell} } \\
    &\approx \int_{-T}^T \Ai(x+v) \Ai(y+v) \frac{e^{\alpha v}}{1+e^{\alpha v}} d v 
    \end{split}
  \end{equation}
and the latter can be made arbitrarily close to the result by choosing $T$ large enough. It remains to show the contributions $|\Sigma_1|, |\Sigma_3| \to 0$. We begin with $\Sigma_1$. By using the same approximation and recognizing a similar Riemann sum, we have
  \begin{equation}
    \begin{split}
    |\Sigma_1| &\approx L^{-1/3} \left| \sum_{\ell=-\infty}^{-T L^{1/3}-1} \left(L^{1/3} J_{2L + x L^{1/3} + \ell} (2 L) \right) \left( L^{1/3} J_{2L + y L^{1/3} + \ell} (2 L)\right) \frac{e^{\alpha L^{-1/3} \ell} } {1 + e^{\alpha L^{-1/3} \ell}}\right| \\
    &\approx L^{-1/3} \left| \sum_{\ell=-\infty}^{-T L^{1/3}-1}  \Ai(x+\ell L^{-1/3}) \Ai(y+\ell L^{-1/3}) \frac{e^{\alpha L^{-1/3} \ell} } {1+e^{\alpha L^{-1/3} \ell} } \right| \\
    &\approx \left| \int_{-\infty}^{-T} \Ai(x+v) \Ai(y+v) \frac{e^{\alpha v}}{1+e^{\alpha v}} d v \right|
    \end{split}
  \end{equation}
  and the latter integral can be made arbitrarily small for appropriately large $T$. For instance, as the Airy function $\Ai$ is bounded, the Airy integral above is bounded by some positive constant times the integral $\int_{-\infty}^{-T} \frac{e^{\alpha v}}{1+e^{\alpha v}} d v = \log(1+\exp(-\alpha T))/\alpha$ which is arbitrarily small if one chooses $T$ large enough. 

  We finally come to $\Sigma_3$ and for a change we shall use the asymptotics in~\eqref{eq:nic2}:
  \begin{equation}
    \begin{split}
    |\Sigma_3| & \approx L^{-1/3} \left| \sum_{\ell = T L^{1/3}+1}^{\infty} \left(L^{1/3} J_{2L + x L^{1/3} + \ell} (2 L) \right) \left( L^{1/3} J_{2L + y L^{1/3} + \ell} (2 L)\right) \frac{e^{\alpha L^{-1/3} \ell} } {1 + e^{\alpha L^{-1/3} \ell}} \right|  \\
    & < L^{-1/3} \left| \sum_{\ell = T L^{1/3}+1}^{\infty} \left(L^{1/3} J_{2L + x L^{1/3} + \ell} (2 L) \right) \left( L^{1/3} J_{2L + y L^{1/3} + \ell} (2 L)\right) \right| \\
    & <  c_1^2 L^{-1/3} \sum_{m = 1}^{\infty} e^{-c_2 ( x + y + 2 T + 2 m L^{-1/3} )} = c_1^{2} e^{-c_2 (x+y+2T)} L^{-1/3} \sum_{m=1}^\infty{q^m} 
    \end{split}
  \end{equation}
  where $q = \exp(-2 c_2 L^{-1/3})$; for the first inequality we have used the crude bound $e^v/(1+e^v) < 1$; for the second we applied~\eqref{eq:nic2} (note we need at least $T \geq \max(-x, -y)$ for this); and for the last equality we performed the change of variables $\ell-T L^{1/3} = m$. The infinite series evaluates to $q/(1-q) \approx L^{1/3}/(2 c_2)$ and hence $|\Sigma_3| < c_3 \exp(-(x+y+2T))$ for an appropriate positive constant $c_3$, and this latter bound can be made arbitrarily small by choosing an appropriately large $T$ . This concludes the proof of \eqref{eq:Kconv}, and the proof of \eqref{eq:Ktrconv} is similar mutatis-mutandis. \qed

\begin{rem}
  The above argument adapts easily to the case $\alpha = \infty$ or more precisely to the case $\alpha L^{-1/3} \to \infty$ as $L \to \infty$. In equation~\eqref{eq:sig_1}, the term $\frac{e^{\alpha L^{-1/3} \ell} } {1 + e^{\alpha L^{-1/3} \ell}}$ becomes the indicator of $\{\ell > 0\}$ and by dominated convergence $\Sigma_2$ tends to the usual (here truncated) Airy kernel $(x, y) \mapsto \int_0^T \Ai(x+v) \Ai(y+v) d v$ while of course we still have $\Sigma_1, \Sigma_3 \to 0$. See also Remark~\ref{rem:bulk} for the case of the bulk limit.
\end{rem}

\paragraph{Proof of Proposition~\ref{prop:Kconv} via the contour integral
representation.}

We assume here that $\alpha<\infty$, see Remark~\ref{rem:alphainfty}
below for the case $\alpha=\infty$. We start with a useful lemma
regarding the behavior of the propagator $\kappa(z,w)$.

\begin{lem}
  \label{lem:kapbound}
  Let $u= e^{-r}$ and $z/w=e^{r/2+\im \theta}$ with $\theta \in [-\pi,\pi]$.
  Then we have
  \begin{equation}
    \kappa(z,w) = \frac{\pi}{r \cosh \frac{\pi \theta}{r}} + O(e^{-\pi^2/r}/r), \qquad r \to 0^+
  \end{equation}
  where the big $O$ is uniform over $\theta$.
\end{lem}

\begin{proof}
  This immediately follows from the Poisson summation formula
  \begin{equation}
    \kappa(z,w) = \sum_{m \in \Z'} \frac{e^{\im \theta m}}{2\cosh \frac{r m}{2}}
    = \sum_{k \in \Z} (-1)^k \int_{-\infty}^\infty  \frac{e^{\im \theta x- 2 \im \pi k x}}{2\cosh \frac{r x}{2}} d x = \sum_{k \in \Z} (-1)^k \frac{\pi}{r \cosh \frac{\pi (\theta-2\pi k)}{r}}.
  \end{equation}
\end{proof}

\begin{rem}
  Lemma~\ref{lem:kapbound} is very similar
  to~\cite[Proposition~3.2]{B2007cyl}, which was proved using the
  elliptic form~\eqref{eq:kapprod} and the imaginary Jacobi transform,
  which is another consequence of the Poisson summation formula. It
  seems that our proof is much simpler.
\end{rem}

We now rewrite the contour integral representation~\eqref{eq:Kcontour}
in the form
  \begin{equation}
    \label{eq:Kplanchbis} L^{1/3} K(a,b) = \frac{1}{(2\im\pi)^2}
\oiint e^{L (S(z) - S(w))} \kappa(z,w) \frac{L^{1/3}\, d z\,
d w}{z^{a'+1} w^{-b'+1}}
\end{equation} with $S(z):=z-z^{-1}-2 \ln z$, $a':=a-2L$,
$b':=b-2L$, and $z$ and $w$ are integrated over the circles of
respective radii $e^{r/4}$ and $e^{-r/4}$, where we set again
$u=e^{-r}$ with $r \to 0^+$. Note that, in the scaling regime
\eqref{eq:edgescal}, we have $r L^{1/3} \to \alpha$, $r a' \to \alpha
x$ and $r b' \to \alpha y$.

  Let us first record some useful properties of the ``action'' $S$:
writing $z=e^{r/4+\im \varphi}$ we have
  \begin{equation}
    \label{eq:Sbound} \Re S(z) = 2 \sinh \left(\frac{r}{4}\right) \cos
\varphi - \frac{r}{2}
  \end{equation} which is maximal when $\varphi=0$, i.e.\ when $z$ is
on the positive real axis. By the relation $S(w)=-S(w^{-1})$, $-\Re
S(w)$ is also maximal on the positive real axis, and since $S(e^{r/4})
\sim r^3/196$, the exponential factor $e^{L (S(z) - S(w))}$ remains
uniformly bounded.

  Fixing an $\epsilon \in (0,1/4)$, we split the double
integral~\eqref{eq:Kplanchbis} in two: let $I$ be the contribution
from the $z$ and $w$ whose arguments are both smaller than
$r^{1-\epsilon}$ in absolute value, and let $I^c$ be the complementary
contribution. We shall show that $I \to M_\alpha(x,y)$ while $I^c \to
0$: 
  \begin{itemize} 
  \item \emph{$I^c$ tends to $0$:} we bound the integrand as
follows. Since either $z$ or $w$ has an argument larger than
$r^{1-\epsilon}$ in absolute value, we have by \eqref{eq:Sbound}
  \begin{equation} \Re S(z)-\Re S(w) - 2 S(e^{r/4}) \leq 2 \sinh
\left(\frac{r}{4}\right) \left(\cos r^{1-\epsilon} -1\right) \leq
-\frac{1}{4} \sinh \left(\frac{r}{4}\right) r^{2-2\epsilon}
  \end{equation} from which we deduce that $|e^{L (S(z) - S(w))}| \leq
C e^{-\alpha^3 r^{-2\epsilon}/8}$ for some $C>0$. By
Lemma~\ref{lem:kapbound}, we have $\kappa(z,w)=O(r^{-1})$ uniformly,
while $|z^{a'} w^{-b'}| = e^{r(a'+b')/4}$ remains bounded away from
$0$, so the integrand tends uniformly to $0$ as wanted.

\item \emph{$I$ converges:} we perform the change of
  variables
  \begin{equation} \label{eq:Ichvar}
    z = e^{r \zeta/\alpha}, \qquad w = e^{r \omega/\alpha}
  \end{equation} where $\Re(\zeta)=-\Re(\omega)=\alpha/4$ while
  $\Im(\zeta),\Im(\omega)$ are both smaller than $\alpha r^{-\epsilon}$
  in absolute value. By a Taylor expansion of $S(z)$ around the ``monkey saddle'' $z=1$, we find that
  \begin{equation}
    e^{L(S(z)-S(w))} = e^{\frac{\zeta^3}{3} - \frac{\omega^3}{3} + O(r^{1-4\epsilon})}
  \end{equation}
  where the big $O$ is uniform. By Lemma~\ref{lem:kapbound}, we have
  \begin{equation}
    \label{eq:kapequiv}
    \kappa(z,w) \sim \frac{\pi}{r \cosh \frac{\pi
        \Im(\zeta-\omega)}{\alpha}} = \frac{\pi}{r \sin \frac{\pi
        (\zeta-\omega)}{\alpha}}
  \end{equation}
  with a uniform exponentially small error term. Finally,
  $z^{a'} w^{-b'} \to e^{x \zeta-y \omega}$, which allows to conclude
  by dominated convergence that
  \begin{equation}
    \label{eq:Iconv}
    I \to \frac{1}{(2\im \pi)^2} \int_{\alpha/4+\im \R}
    d \zeta \int_{-\alpha/4+\im \R} d \omega \exp \left(
      \frac{\zeta^3}{3} - \frac{\omega^3}{3} - x \zeta + y \omega
    \right) \frac{\pi}{\alpha \sin \frac{\pi (\zeta-\omega)}{\alpha}}
  \end{equation}
  (note that the $L^{1/3}$ of \eqref{eq:Kplanchbis} and the $1/r$ of
  \eqref{eq:kapequiv} are absorbed by the Jacobian of the change of
  variables $z\to \zeta$, $w \to \omega$).
\end{itemize}
It is straightforward to identify the right-hand side of
\eqref{eq:Iconv} as $M_\alpha(x,y)$, using the integral
representations
\begin{equation} \frac{\pi}{\alpha \sin \frac{\pi
      (\zeta-\omega)}{\alpha}} = \int_{-\infty}^\infty
  \frac{e^{(\alpha+\omega-\zeta) v} d v}{1+e^{\alpha v}}, \qquad
  \Ai(x+v) = \int_{\alpha/4+\im \R} e^{\frac{\zeta^3}{3} - (x+v)
    \zeta} \frac{d \zeta}{2 \im \pi}
\end{equation}
and similarly for $\Ai(y+v)$ which is given by the integral over
$\omega$. This concludes the proof of~\eqref{eq:Kconv}.

The proof of the convergence of traces~\eqref{eq:Ktrconv} is entirely
similar: in~\eqref{eq:Kplanchbis} the rightmost fraction should simply
be replaced by $\frac{dz dw}{(z-w)z^m w^{-m+1}}$ (the factor
$\frac{1}{z-w}$ is of order $r^{-1} \sim L^{1/3}$, and produces after
the change of variables a factor $\frac{1}{\zeta-\omega}$ which we
``decouple'' by the usual integral representation
$\frac{1}{\zeta-\omega}=\int_0^\infty e^{(\omega-\zeta)v'}d v'$).

\qed

\begin{rem}
  \label{rem:alphainfty}
  In the case $\alpha=\infty$ the above proof may be adapted
  straightforwardly by integrating $z$ and $w$ over the circles of
  radii $e^{\pm L^{-1/3}}$ instead of $e^{\pm r/4}$. The change of
  variable~\eqref{eq:Ichvar} is simply replaced by
  $z=e^{L^{-1/3} \zeta}$ and $w=e^{L^{1/3} \omega}$ with
  $\Re(\zeta)=-\Re(\omega)=1$ and, using the Poisson summation formula
  as for Lemma~\ref{lem:kapbound}, we see that $\kappa(z,w)$ may now
  be approximated by $\frac{L^{1/3}}{\zeta-\omega}$ over the
  integration contour, instead of~\eqref{eq:kapequiv}. We recover the
  usual zero-temperature Airy kernel as the limit.
\end{rem}

\begin{rem}
  \label{rem:universality}
  As observed in several occasions---see e.g.\ \cite{Oko02,or}---the
  asymptotic analysis of contour integral representations is
  particularly robust and explains why the limiting behavior should be
  universal. In our analysis, we simply use the fact that the action
  $S(z)$ has a double critical point on the real axis, a property that
  is characteristic of the ``edge'' regime. Our proof should therefore
  be easily adaptable to other instances of periodic Schur
  processes. For example, we expect that the finite-temperature Airy
  kernel will be observed in cylindric partitions as the scaling limit
  around any generic edge point (the action now involving the
  dilogarithm function~\cite{or,B2007cyl}).
\end{rem}

\section{The stationary cylindric Plancherel process}
\label{sec:cylPlanchproc}

In this section we define a continuous-time periodic extension of the
cylindric Plancherel measure. This process has the remarkable property
of being stationary, i.e.\ invariant under time translation. It may be
identified as the periodic analogue of the stationary process of
Borodin and Olshanski~\cite{BO06}, whose fixed-time marginal is the
ordinary poissonized Plancherel measure. As we constructed our process
through Fock space considerations, our presentation is somewhat
different from~\cite{BO06}.

\paragraph{Definition and basic properties.}

For $\beta,\vartheta$ two nonnegative parameters, let $\Tt(\beta)$ be
the ``transfer matrix'' acting on the set $\Par$ of integer partitions
by
\begin{equation}
  \label{eq:Ttdef}
  \Tt(\beta)_{\lambda,\nu} := e^{\vartheta^2(e^{-\beta}-1)}
  \sum_{\mu \in \Par} e^{-\beta|\mu|} s_{\lambda/\mu}(\ex_{\vartheta(1-e^{-\beta})})
  s_{\nu/\mu}(\ex_{\vartheta(1-e^{-\beta})}), \qquad \lambda,\nu \in \Par
\end{equation}
where we recall that $\ex$ denotes the exponential specialization
\eqref{eq:expspec}---we may therefore rewrite
$\Tt(\beta)_{\lambda,\nu}$ in terms of the numbers of
standard Young tableaux $\dim(\lambda/\mu)$ and
$\dim(\nu/\mu)$. What is remarkable about our definition
is that the family $\left(\Tt(\beta)\right)_{\beta \geq 0}$ forms a
semigroup.

\begin{prop}
  \label{prop:semigroup}
  For any $\beta,\beta',\vartheta$ we have
  \begin{equation}
    \label{eq:semigroup}
    \Tt(\beta) \Tt(\beta') = \Tt(\beta+\beta').
  \end{equation}
  Furthermore for any $\vartheta$ we have
  $\sum_{\lambda \in \Par} \Tt(\beta)_{\lambda,\lambda}=Z_\beta$,
  where $Z_\beta:=\prod_{k \geq 1}(1-e^{-\beta k})^{-1}$.
\end{prop}

\begin{proof}
  The transfer matrix is represented in fermionic Fock space as
  \begin{equation}
    \label{eq:TtFock}
    \Tt(\beta) = e^{\vartheta^2(e^{-\beta}-1)} 
    \Gamma_-(\ex_{\vartheta(1-e^{-\beta})}) e^{-\beta H} \Gamma_+(\ex_{\vartheta(1-e^{-\beta})})
  \end{equation}
  (we keep the same letter by a slight abuse of notation). We then
  leave to the reader the pleasure of checking that the commutation
  relations~\eqref{eq:gamcomm} imply the desired relation. The trace
  formula for $Z_\beta=\tr(\Pi_0 \Tt(\beta))$ may be checked using the
  $\Gamma$-elimination method of Section~\ref{sec:fermions}.
\end{proof}

\begin{rem}
  \label{rem:Ttnormord}
  In fact, we have $\Tt(\beta)=e^{-\beta H_\vartheta}$ where
  $H_\vartheta:=H-\vartheta(\alpha_1+\alpha_{-1})+\vartheta^2$ with
  $\alpha_{\pm 1}$ being the bosonic operators defined
  in~\eqref{eq:gamalpdef}. The expression~\eqref{eq:TtFock}
  corresponds to the normal ordered form.
\end{rem}

\begin{definition}
  \label{def:cylPlanchproc}
  Fix $\beta,\vartheta$ two nonnegative real parameters. The
  (stationary) \emph{cylindric Plancherel process} of period $\beta$
  and intensity $\vartheta$ is the random partition-valued
  continuous-time process $\lambda(\cdot): \R \to \Par$ which is
  $\beta$-periodic and whose finite-dimensional distributions are
  given by
  \begin{equation}
    \label{eq:cylPlanchproc}
    \Prob(\lambda(0),\lambda(b_1),\ldots,\lambda(b_n))
    = \frac{1}{Z_\beta} \prod_{i=0}^n
    \Tt(b_{i+1}-b_i)_{\lambda(b_i),\lambda(b_{i+1})}
  \end{equation}
  where $0=b_0<b_1<\cdots<b_n<b_{n+1}=\beta$.
\end{definition}

By Proposition~\ref{prop:semigroup} it is clear that
\eqref{eq:cylPlanchproc} defines a consistent family of
finite-dimensional distributions, so the process is well-defined by
the Kolmogorov extension theorem. We may see that $\lambda(\cdot)$ can
be realized as a partition-valued jump process (with finitely many
jumps almost surely)---see again Figure~\ref{fig:CPP_paths} for an
illustration---but we shall not detail these considerations here as we
are chiefly interested in the properties of the finite-dimensional
distributions themselves.  Note that the operator $-H_\vartheta$
mentioned in Remark~\ref{rem:Ttnormord}, while having positive
off-diagonal elements in the canonical basis, is not an intensity
matrix so does not directly generate a Markov process. However,
$H_\vartheta$ is self-adjoint and nonnegative, hence our process might
be thought as a quantum evolution in imaginary time.

The law of $\lambda(0)$ is the $\lambda$-marginal of the cylindric
Plancherel measure~\eqref{eq:uPlanchdef} of parameters $u=e^{-\beta}$
and $\gamma=\vartheta(1-e^{-\beta})$. This is also true for
$\lambda(b)$ with any $b \in \R$ as, from the definition and from the
semigroup property, one may check easily that the process is
stationary. Moreover, from~\eqref{eq:Ttdef} we may check that, for any
times $0 \leq b_1<\cdots<b_N \leq \beta$, the tuple
$(\lambda(b_1),\ldots,\lambda(b_N))$ is the $\lambda$-marginal of a
periodic Schur process as defined in the introduction. Note that this
requires to rewrite the weight~\eqref{eq:cylPlanchproc} in the
form~\eqref{eq:probdef}: one sees that $u$ is always $e^{-\beta}$ and
that the specializations $\rho^\pm_k$ will all be exponential, but
their parameters have a slightly unnatural expression---which we omit
here---since \eqref{eq:probdef} breaks translation invariance due to
the weight $u^{|\mu^{(0)}|}$.

\begin{rem}[Infinite period/zero-temperature limit]
  In Fock space representation we have
  \begin{equation}
    \lim_{\beta \to \infty} \Tt(\beta) = e^{-\vartheta^2} 
    \Gamma_-(\ex_{\vartheta}) \vv \vcv \Gamma_+(\ex_{\vartheta}).
  \end{equation}
  In other words $\Tt(\infty)$ has rank one and we have the
  factorization
  $\Tt(\infty)_{\lambda,\nu}= e^{-\vartheta^2}
  s_\lambda(\ex_\vartheta) s_\nu(\ex_\vartheta)$. From this
  observation, we see that the cylindric Plancherel process has
  well-defined limit as $\beta \to \infty$, which is an ordinary
  (nonperiodic) continuous-time Schur process whose marginal at any
  given time is the poissonized Plancherel measure of parameter $\vartheta$. We
  may identify it with the stationary process of Borodin and
  Olshanski~\cite{BO06}: this is not obvious from the definition, but
  it will be from the expression of the correlation functions.
\end{rem}

\begin{rem}
  In view of the characterization~\eqref{eq:rhoposparam} for
  nonnegative specializations, the exponential specialization is the
  only one having the ``infinite divisibility property'' needed to
  define a stationary continuous-time Schur process. Therefore the
  cylindric Plancherel process seems to be one of a
  kind. In~\cite{BO06z}, Borodin and Olshanski construct Markov
  processes preserving $z$-measures but, as they point out, they do
  not seem to belong to the class of Schur process.
\end{rem}

\begin{rem}
  The cylindric Plancherel process appears as a ``poissonian'' limit
  of a measure on cylindric partitions in the following sense: let us
  consider the instance of the periodic Schur process of length $N$
  given by
  \begin{equation}
    \label{eq:probdefcp}
    \Prob(\vec{\lambda},\vec{\mu})\ \propto\ 
    \prod_{k=1}^N \left(s_{\lambda^{(k)}/\mu^{(k-1)}}\left(\vartheta \epsilon\right) e^{-\epsilon |\mu^{(k)}|}
      s_{\lambda^{(k)}/\mu^{(k)}}\left(\vartheta \epsilon\right) \right).
  \end{equation}
  This weight may be rewritten in the form~\eqref{eq:probdef} but here
  translation invariance is manifest. Since the specializations
  consist of single variables, the weight is supported on sequences
  $(\vec{\lambda},\vec{\mu})$ corresponding to cylindric partitions,
  see again Figure~\ref{fig:cylpp2}. Then, the cylindric Plancherel
  measure of period $\beta$ and intensity $\vartheta$ appears as the
  limit $N \to \infty$, $\epsilon \to 0^+$ with $N \epsilon \to
  \beta$. This is because the ``elementary transfer matrix'' reads
  here
  $\Gamma_-(\vartheta \epsilon) e^{-\epsilon H} \Gamma_+(\vartheta
  \epsilon) = 1 - \epsilon (H_\vartheta-\vartheta^2)+O(\epsilon^2)$ in
  the notation of Remark~\ref{rem:Ttnormord}.
\end{rem}

\paragraph{Correlation functions.}

We now turn to correlation functions, which characterize the point
process
\begin{equation}
  \mathfrak{S}(\lambda):= \left\{ \left(b,\lambda(b)_i-i+\frac{1}{2}\right),
    \ b\in [0,\beta], \ i\geq 1\right\} \subset [0,\beta] \times \Z'.
\end{equation}
Following the discussion of Section~\ref{sec:corr}, we shall rather
consider the shift-mixed process $\mathfrak{S}(\lambda)+(0,c)$, where
$c \in \Z$ is distributed as~\eqref{eq:cdist} with $t$ an arbitrary
positive parameter.

\begin{prop}
  The shift-mixed process $\mathfrak{S}(\lambda)+(0,c)$ is
  determinantal with correlation kernel given by
  \begin{equation}
    \label{eq:Kext}
    K(b,k;b',k') = \frac{1}{(2\im \pi)^2} \oint_{|z|=1^+} \frac{d z}{z^{k+1}} \oint_{|w|=1^-} \frac{d w}{w^{-k'+1}}  \frac{e^{\vartheta(z-z^{-1})}}{e^{\vartheta(w-w^{-1})}}  \kappa(z e^{-b},w e^{-b'})
  \end{equation}
  where: $b,b' \in [0,\beta)$, $k,k' \in \Z'$; $\kappa$ is as in
  \eqref{eq:kapform} with $u=e^{-\beta}$; and the notation $1^\pm$
  means that we should take $|z|$ a bit larger than $|w|$ in the case
  $b=b'$, to circumvent the pole of $\kappa$ at $z=w$.
\end{prop}

\begin{proof}
  That the shift-mixed process is determinantal is a consequence of
  Theorem~\ref{thm:main_thm}. We may also derive the expression of
  $K(b,k;b',k')$ by converting the weight~\eqref{eq:cylPlanchproc} in
  the form~\eqref{eq:probdef} and doing further manipulations, but
  we found it more instructive and less error-prone to rederive it
  directly from the free fermion formalism. From the discussion of
  Section~\ref{sec:fermions}, the generating function
  \begin{equation}
    \mathcal{K}(b,z;b',w) := \sum_{k \in \Z'} \sum_{k' \in \Z'}
    K(b,k;b',k') z^k w^{-k'}
  \end{equation}
  admits the Fock space representation
  \begin{equation}
    \label{eq:KprocTt}
    \mathcal{K}(b,z;b',w) = \mathcal{Z}_{\beta,t}^{-1} \times
    \begin{cases}
       \tr \left( t^C \Tt(b) \psi(z) \Tt(b'-b) \psi^*(w)
        \Tt(\beta-b') \right) & \text{if $b \leq b'$,}\\
      - \tr \left( t^C \Tt(b') \psi^*(w) \Tt(b-b') \psi(z)
        \Tt(\beta-b) \right) & \text{if $b > b'$.}
    \end{cases}
  \end{equation}
  where $\mathcal{Z}_{\beta,t}=\tr(t^C \Tt(\beta))=\tr(t^C e^{-\beta H})$.
  We now plug in \eqref{eq:TtFock} and perform $\Gamma$-elimination,
  which yields
  \begin{equation}
    \label{eq:KprocbH}
    \mathcal{K}(b,z;b',w) = \mathcal{Z}_{\beta,t}^{-1} \frac{e^{\vartheta(z-z^{-1})}}{e^{\vartheta(w-w^{-1})}} \times
    \begin{cases}
      \tr \left( t^C e^{-bH} \psi(z) e^{-(b'-b)H} \psi^*(w)
        e^{-(\beta-b')H} \right) & \text{if $b \leq b'$,}\\
      - \tr \left( t^C e^{-b'H} \psi^*(w) e^{-(b-b')H} \psi(z)
        e^{-(\beta-b)H} \right) & \text{if $b > b'$.}
    \end{cases}
  \end{equation}
  Here, the somewhat surprising fact that the prefactor is independent
  of $b$ and $b'$ can be explained as follows: observe that, when
  performing $\Gamma$-elimination, the factor produced by the
  quasi-commutations with a given fermionic operator, say $\psi(z)$,
  does not depend on the position of the other fermionic operator(s)
  in the product. Therefore this factor should be the same as for
  $b=b'$, and by using translation invariance we may reduce to the
  case $b=b'=0$, from which we may identify the factor with the one
  present in~\eqref{eq:Kcontour} for the cylindric Plancherel measure
  (note that $L=\vartheta$ in our current notations).

  Finally, from the relations
  $e^{-b H} \psi(z) e^{b H} = \psi(z e^{-b})$ and
  $e^{-b' H} \psi^*(w) e^{b' H} = \psi^*(w e^{-b'})$ (which result
  from the fact that $\psi_k$ and $\psi^*_{-k}$ ``increase the
  energy'' by $k$) and from the fermionic expression~\eqref{eq:kferm}
  of $\kappa$, we evaluate the remaining traces in \eqref{eq:KprocbH}
  and get
  \begin{equation}
    \mathcal{K}(b,z;b',w) = \frac{e^{\vartheta(z-z^{-1})}}{e^{\vartheta(w-w^{-1})}} \kappa(z e^{-b},w e^{-b'}).
  \end{equation}
  The ratio of the arguments of $\kappa$ should be in the annulus
  $(1,e^\beta)$ for $b \leq b'$ and in $(e^{-\beta},1)$ for $b > b'$,
  and this condition is ensured by taking $|z|=1^+$ and $|w|=1^-$. The
  coefficient $K(b,k;b',k')$ may then be extracted by a double contour
  integral.
\end{proof}

\begin{rem}
  The Bessel representation~\eqref{eq:KBessel} of the correlation
  kernel for the cylindric Plancherel measure may be generalized as
  \begin{equation}
    K(b,k;b',k') =  
    \begin{cases}
      \sum_{\ell \in \Z'} J_{k+\ell}(2\vartheta)
      J_{k'+\ell}(2\vartheta) \frac{e^{(b-b')\ell}}{1+t^{-1}
        e^{\beta \ell}} & \text{if $b \leq b'$,}\\
      - \sum_{\ell \in \Z'} J_{k+\ell}(2\vartheta)
      J_{k'+\ell}(2\vartheta) \frac{e^{(b-b')\ell}}{1+t
        e^{-\beta \ell}} & \text{if $b >b'$.}\\
  \end{cases}
  \end{equation}
  In the limit $\beta \to \infty$, we recover the stationary version
  of the discrete extended Bessel kernel of~\cite[Theorem~3.3]{BO06}.
\end{rem}

\paragraph{Edge crossover behavior.} We now study the edge crossover asymptotics of the
kernel $K(b,k;b',k')$, which we find to be described by the
\emph{finite-temperature extended Airy kernel}, as defined
in~\cite[Equation~(99)]{dms}. Proposition~\ref{prop:Kconv} admits the
following extension. (Note that the convergence of traces need not be
extended.)
\begin{prop}
  \label{prop:Kextconv}
  For $t=1$ and in the edge crossover regime
  \begin{equation}
    \label{eq:extedgescal}
   \left\{
      \begin{array}[c]{l}
        \vartheta \to \infty \\ \beta \to 0^+ \\
        \vartheta^{1/3} \beta \to \alpha
      \end{array} \right.
    \qquad
    \left\{
      \begin{array}[c]{l}
        b = \beta \tau/\alpha \\
        k=\lfloor 2\vartheta+x\vartheta^{1/3} \rfloor \\
        b' = \beta \tau'/\alpha \\
        k'=\lfloor 2\vartheta+y\vartheta^{1/3} \rfloor \\
     \end{array}
   \right. 
  \end{equation}
  where we assume $\tau,\tau' \in [0,\alpha)$ without loss of generality, we have
  \begin{equation}
    \label{eq:Kextconv}
    \vartheta^{1/3} K(b,k;b',k') \to
    \begin{cases}
    \int_{-\infty}^\infty \frac{e^{(\tau-\tau') v}}{1+e^{-\alpha v}} \Ai(x+v) \Ai(y+v) d v & \text{if $\tau \leq \tau'$,}\\
    - \int_{-\infty}^\infty \frac{e^{(\tau-\tau') v}}{1+e^{\alpha v}} \Ai(x+v) \Ai(y+v) d v  & \text{if $\tau > \tau'$}.
  \end{cases}
  \end{equation}
\end{prop}

\begin{proof}
  We simply adapt the contour integral proof of Proposition~\ref{prop:Kconv}. Let us highlight the few
  changes. First, obviously, the role of $r$ is now played by $\beta$,
  and that of $L$ by $\vartheta$. An immediate generalization of
  Lemma~\ref{lem:kapbound} shows that $\kappa=O(\beta^{-1})$ uniformly
  on the integration contours, and we have the estimate
  \begin{equation}
    \label{eq:kapboundext}
    \kappa(z e^{-b},w e^{-b'}) \sim \frac{\pi}{\beta \sin \frac{\pi
        (\zeta-\omega+\tau'-\tau)}{\alpha}}
  \end{equation}
  in the vicinity of the saddle-point, where we perform the change of
  variables $z=e^{\beta \zeta/\alpha}$, $w=e^{\beta
    \omega/\alpha}$. Here, $\zeta$ and $\omega$ shall be integrated
  over lines parallel to the imaginary axis, with
  $\Re(\zeta)>0>\Re(\omega)$ and with $\Re(\zeta-\omega+\tau'-\tau)>0$
  if $\tau \leq \tau'$, $<0$ otherwise. In view of~\eqref{eq:Kext},
  all the dependence in $b,b'$ being in the $\kappa$ factor, we
  conclude that the limiting kernel is given by~\eqref{eq:Iconv} where
  we replace the last factor by the right-hand side
  of~\eqref{eq:kapboundext}. The last step is to plug in the integral
  representations of the cosecant and Airy functions, but for the
  cosecant we should be careful that it depends on the sign of
  $\tau'-\tau$, namely
  \begin{equation}
    \frac{\pi}{\alpha \sin \frac{\pi(\zeta-\omega+\tau'-\tau)}{\alpha}} =
    \begin{cases}
      \int_{-\infty}^\infty
      \frac{e^{-(\zeta-\omega+\tau'-\tau) v} d v}{1+e^{-\alpha v}} & \text{if $\tau \leq \tau'$,}\\
      - \int_{-\infty}^\infty
      \frac{e^{-(\zeta-\omega+\tau'-\tau) v} d v}{1+e^{\alpha v}} & \text{if $\tau > \tau'$.}\\
    \end{cases}
  \end{equation}
  This leads to the desired expression~\eqref{eq:Kextconv}.
\end{proof}

Of course, some extra steps of analysis---which we omit---are needed
to deduce the convergence of the rescaled process from that of the
kernel. Note that Lemma~\ref{lem:Cestim} still ensures that
shift-mixing is irrelevant in the considered scaling limit.  We have
not investigated the edge high temperature regime, namely whether
Proposition~\ref{prop:Kconvgum} may be extended to the cylindric
Plancherel process.

\section{The periodic strict Schur process and neutral fermions}
\label{sec:corr_S}

Vuleti\'c~\cite{vul2} defined the so-called shifted Schur process,
which is a measure on sequences of \emph{strict} partitions (as
defined in Section~\ref{sec:fermidirac}), and whose definition
involves Schur's $P$ and $Q$ functions instead of the ordinary Schur
functions. In this section we explore the periodic variant of the
shifted Schur process which we prefer to call the \emph{strict Schur
process} to emphasize that it deals with strict partitions.

\paragraph{Reminders on Schur's $P$ and $Q$ functions.}

(See \cite[Section III.8]{mac} for more general background.) Let
$\mathsf{SSym} \subset \mathsf{Sym}$ be the subalgebra of symmetric
functions generated by the $p_n$'s with $n$ odd. This subalgebra
contains the $q_n$'s which are defined via the generating series
\begin{equation}
  Q(z) := \sum_{n \geq 0} q_n(x_1, x_2, \dots) z^n = \prod_{i} \frac{1+z x_i}{1-z x_i}.
\end{equation}
This results from the relation
$Q(z) = \exp(\sum_{n\geq 1,\ \mathrm{odd}} 2 p_n z^n/n)$ and, in fact,
the $q_n$'s with $n$ odd form another algebraically independent family
generating $\mathsf{SSym}$.

For $a<b$ distinct nonnegative integers let $Q_{(a,b)} := q_a q_b + 2 \sum_{i=1}^b (-1)^i q_{a+i} q_{b-i}$ and set $Q_{(b, a)} = - Q_{(a,b)}$. For strict partitions $\mu, \lambda$, write them such that $\lambda_1 > \dots > \lambda_m > 0$ and $\mu_1 > \dots > \mu_n \geq 0 $ so we can assume without loss of generality that $m+n$ is even. Then the skew Schur's $Q$ function $Q_{\lambda / \mu}$ is defined by
$Q_{\lambda / \mu} = \pf M_{\lambda, \mu}$ where $M_{\lambda, \mu}$ is the $(m+n) \times (m+n)$ antisymmetric matrix block-defined by
\begin{equation}
  M_{\lambda, \mu} = \left(
\begin{array}{ll}
M_{\lambda} & N_{\lambda, \mu} \\
-N_{\mu, \lambda} & 0
\end{array}
  \right)
\end{equation}
with $(M_{\lambda})_{i,j} = (Q_{(\lambda_i, \lambda_j)})$ and $N_{\lambda, \mu}$ the $m\times n$ matrix $(q_{\lambda_i - \mu_j})$. This should be viewed as an analogue of the Jacobi--Trudi formula for Schur's Q functions. We define the companion Schur's $P$-functions by $Q_{\lambda / \mu} = 2^{\ell(\lambda)-\ell(\mu)} P_{\lambda / \mu}$. We note that $Q_{\lambda / \mu} = 0$ unless $\mu \subset \lambda$ and likewise for $P$.

A \emph{strict specialization} $\rho$ is an algebra homomorphism from $\mathsf{SSym}$ to the complex numbers, and is completely determined by the generating function
\begin{equation}
  Q(\rho; z) := \sum_{n \geq 0} q_n(\rho) z^n = \exp \left( \sum_{n\geq 1,\ \mathrm{odd}} \frac{2 p_n(\rho) z^n}{n} \right).
\end{equation}
It is said to be nonnegative if $Q_{\lambda / \mu}(\rho) \geq 0$ for all $\lambda, \mu$. 
A necessary and sufficient condition for $\rho$ to be nonnegative is \cite{naz1, naz2} that its generating function is of the form
\begin{equation}
  \label{eq:rho_S_posparam}
  Q(\rho; z) = e^{\gamma z} \prod_{i \geq 1} \frac{1+\alpha_i z}{1-\alpha_i z}
\end{equation}
for a summable sequence of nonnegative real numbers
$\gamma, \alpha_1, \alpha_2, \dots$. For two strict specializations
$\rho,\rho'$, we define the specialization $\rho \cup \rho'$ by
\begin{equation}
  Q(\rho \cup \rho';z) := Q(\rho; z) Q(\rho'; z).
\end{equation}
Recalling that $\SPar$ denotes the set of strict partitions, we can
finally define $Q(\rho; \rho')$ by
\begin{equation}
  \label{eq:hdef_S}
  Q(\rho; \rho') := \sum_{\lambda \in \SPar} Q_\lambda(\rho) P_\lambda(\rho') = \sum_{\lambda \in \SPar} 2^{\ell(\lambda)} P_\lambda(\rho) P_\lambda(\rho').
\end{equation}
This is associated to the Cauchy identity for Schur's $P$ and $Q$ functions \cite[Section III.8]{mac},
which amounts to
\begin{equation}
  \label{eq:cauchy_S}
  Q(\rho; \rho') = \exp \left( \sum_{n \geq 1,\text{ odd}} \frac{2 p_n(\rho) p_n(\rho')}{n} \right).
\end{equation}

\paragraph{Definition of the process and the partition function.}

Fix $N$ a nonnegative integer, $u$ a nonnegative real number smaller
than $1$, and let $\rho^{\pm}_i,\ 1 \leq i \leq N$ be $2N$ strict
specializations such that $z \mapsto Q(\rho^{\pm}_i;z)$ is analytic in
a disk of radius $R>1$. The \emph{periodic strict Schur process} is a
measure on sequences of \emph{strict} partitions of the
form~\eqref{eq:seqform}, which to a given sequence
assigns the weight
\begin{equation}
  \label{eq:wt_def_S}
  \mathcal{W}_s(\vec{\lambda},\vec{\mu}) := u^{|\mu^{(0)}|} 
  \prod_{k=1}^N \left(Q_{\lambda^{(k)}/\mu^{(k-1)}}\left(\rho_k^+\right)
    P_{\lambda^{(k)}/\mu^{(k)}}\left(\rho_k^-\right) \right).
\end{equation}
The \emph{partition function}
$Z_s \equiv Z_s(u, \rho_1^+, \rho_1^-,\ldots,\rho_N^+,\rho_N^-)$ is
the sum of weights of all sequences.

\begin{prop}
  \label{prop:fbz_S}
  The partition function of the periodic strict Schur process reads
  \begin{equation}
    \label{eq:fbz_S}
    Z_s = \prod_{1 \leq k \leq \ell \leq N} Q(\rho_k^+;\rho_\ell^-)
    \prod_{n \geq 1} Q(u^{n}\rho^+; \rho^-)(1+u^n)
  \end{equation}
  where
  \begin{equation}
    \label{eq:rhopmdef_S}
    \rho^\pm := \rho^\pm_1 \cup \rho^\pm_2 \cup \cdots \cup \rho^\pm_N.
  \end{equation}
\end{prop}

\begin{rem}
  For $u=0$ we recover the partition function of the shifted Schur
  process~\cite{vul2}.
\end{rem}

\paragraph{Correlation functions.}

The \emph{point configuration}
associated with a sample $(\vec{\lambda},\vec{\mu})$ of the periodic strict Schur process is the set
\begin{equation}
  \label{eq:pointconfig_S}
  \mathfrak{S}^s(\vec{\lambda}) := \left\{
    \left(i,\lambda^{(i)}_j \right), 1 \leq i \leq N , 1 \leq j \leq \ell(\lambda^{(i)})
  \right\} \subset \N \times \N^*.
\end{equation}
Note that, since the $\lambda^{(i)}$'s are strict, we need not shift
the parts to obtain a simple point process and, as before, considering
only the partitions $\lambda^{(1)},\ldots,\lambda^{(N)}$ causes no
loss of generality.

Fixing an auxiliary real parameter $t$, let $c$ be a variable which
can take the values $0$ and $1$, and define the \emph{shift-mixed}
periodic strict Schur process as the probability measure
\begin{equation}
  \Prob (c, \vec{\lambda}, \vec{\mu}) = \frac{1}{(1+t) Z_s} \cdot t^c \cdot \mathcal{W}_s (\vec{\lambda}, \vec{\mu}).
\end{equation}
In other words, $c$ is Bernoulli random variable of parameter $t/(1+t)$ independent 
of $(\vec{\lambda}, \vec{\mu})$. The \emph{shift-mixed}
point configuration $\mathfrak{S}^s_t(\vec{\lambda})$ \emph{is the
  same as the point configuration}, except viewed in the larger space
$\{0, 1 \} \times \N \times \N^*$.\footnote{The fiber $c \times \mathfrak{S}^s(\vec{\lambda})$---$c \in \{0, 1\}$---corresponds to states $\ket{\lambda, c}$ in neutral Fock space $\mathcal{NF}$ in the notation of Appendix~\ref{sec:fermions_S}.} In this regard, $t$ is a dummy
parameter here more so than it was in the case of the shift-mixed
periodic (non strict) Schur process, where the dependence of $t$ and
$u$ was intertwined in the theta function $\theta_3(t; u)$. We
nevertheless choose to keep it for keeping our exposition parallel to
that of the case of the non strict process.
We now state our main result.
\begin{thm}
  \label{thm:main_thm_S}
  Let $c \in \{0, 1\}$ be a Bernoulli random variable of parameter $t$ and $U=\{(i_1, k_1),\ldots,(i_n, k_n)\} \subset \{1, \ldots, N\} \times \N^*$ with $i_1 \leq \cdots \leq i_n$. The $n$-point shift-mixed correlation function 
  $\tilde{\varrho}^s(U) := \Prob(c \times U \subset \mathfrak{S}^s_t (\vec{\lambda}))$
  is a pfaffian:
  \begin{equation}
    \tilde{\varrho}^s(U) = \pf K
  \end{equation}
  where $K$ is the $2n \times 2n$ antisymmetric matrix given by
  \begin{equation}
    K_{\gamma, \delta} = 
    \begin{cases}
      2^{-1} \left[ z^{k_{\gamma}} w^{k_{\delta}} \right] F^s(i_{\gamma}, z) F^s(i_{\delta}, w) \kappa^s(z, w), &1 \leq \gamma < \delta \leq n,\\
      2^{-1} \left[ z^{k_{\gamma}} w^{-k_{2n+1-\delta}} \right] (-1)^{k_{2n+1-\delta}} F^s(i_{\gamma}, z) F^s(i_{\delta}, w) \kappa^s(z, w), &1 \leq \gamma \leq n < \delta \leq 2n,\\
      2^{-1} \left[ z^{-k_{2n+1-\gamma}} w^{-k_{2n+1-\delta}} \right] (-1)^{k_{2n+1-\gamma}+k_{2n+1-\delta}} F^s(i_{\gamma}, z) F^s(i_{\delta}, w) \kappa^s(z, w), &n < \gamma < \delta \leq 2n
    \end{cases}
  \end{equation}
  and 
    \begin{equation} \label{eq:Fkap2b_S}
      \begin{split}
        F^s(i, z) &= \frac{\prod_{1 \leq \ell \leq i} Q(\rho_{\ell}^+;z)}{\prod_{i \leq \ell \leq N} Q(\rho_{\ell}^-;z^{-1})} \prod_{n \geq 1} \prod_{1 \leq \ell \leq N} \frac{Q(u^n \rho_{\ell}^+;z)}{Q(u^n \rho_{\ell}^-;z^{-1})}, \\
        \kappa^s(z,w) &= \frac{\theta_{u} \left( \frac{w}{z} \right)}{\theta_{u} \left( - \frac{w}{z} \right)} \cdot (1+t) (u; u)_{\infty}^2 (-u; u)_{\infty}^{-1}, \qquad u^{1/2} < |w| < |z| < u^{-1/2}.
      \end{split}
    \end{equation}
\end{thm}

\begin{rem}
  If $u=0$ (in which case one can take $t=0$) we recover the pfaffian correlations of the shifted Schur process~\cite{vul2}. See also \cite{mat} for the case of the shifted Schur measure.
\end{rem} 

\begin{rem}
  The coefficient extractions in the kernel $K$ can be replaced by
  double contour integrals similar to those in
  Theorem~\ref{thm:main_thm}.
\end{rem}

\paragraph{Proof of Theorem~\ref{thm:main_thm_S}.}

We again use the formalism of free fermions, but this time in their
``neutral'' flavor. See Appendix~\ref{sec:fermions_S} for our
conventions and notations. Let us first sketch the proof of
Proposition~\ref{prop:fbz_S}. We write
\begin{equation}
  Z_s = \tr_0 \left( \usf^{H_s} \Gapl^s(\rho^{+}_1) \Gami^s(\rho_1^-) \dots \Gapl^s(\rho^+_N) \Gami^s(\rho^-_N) \right)
\end{equation}
where $\tr_0$ stands for trace over the subspace $\NF_0$ of even
grading. We then proceed as in Section~\ref{sec:fermions} to eliminate
the $\Gamma$-operators, with the following two minor
modifications. First, commuting the $\Gamma^s$ operators will yield
$Q(u^n \rho^+_i; \rho^-_j)$ factors instead of the $H$
factors. Second, at the end one is left with
$\tr_0 (u^{H_s}) = \prod_{i \geq 1} (1+u^i)$. For the shift-mixed
cylindric strict Schur process, the partition function is computed
similarly except one traces over all of $\NF = \NF_0 \oplus \NF_1$,
the result being $(1+t) Z_s$.

We now turn to the proof of Theorem~\ref{thm:main_thm_S} itself, which
is very similar to that of Theorem~\ref{thm:main_thm}. First we notice that
\begin{equation}
  \begin{split}
  \tilde{\varrho}^s(U) = &\frac{1}{(1+t) Z_s} \tr \left( u^{H_s} t^{C_s} \Gapl^s (\rho_1^+) \Gami^s (\rho_1^-) \cdots \Gapl^s (\rho_{i_1}^+) \left(\frac{1}{2} \phi_{k_1} \phi^*_{k_1}\right) \Gami^s (\rho_{i_1}^-) \cdots \right.\\
  &\left. \Gapl^s (\rho_{i_n}^+) \left( \frac{1}{2} \phi_{k_n} \phi^*_{k_n} \right) \Gami^s (\rho_{i_n}^-) \cdots \Gapl^s (\rho_N^+) \Gami^s (\rho_N^-) \right)
  \end{split}
\end{equation}
where the operators
$\frac{1}{2} \phi_{k_{\gamma}} \phi^*_{k_{\gamma}}$ measure whether
$\lambda^{(i_{\gamma})}$ has a part of size ${k_\gamma}$ and we recall
that $\phi^*_k = (-1)^k \phi_{-k}$ for $k>0$. By commuting out the
$\Gamma$-operators, we can rewrite the correlation function as
\begin{equation}
  \label{eq:rhotdef_S}
  \tilde{\varrho}^s(U) = 2^{-n} \frac{\tr \left( u^{H_s} t^{C_s} \Phi_{k_1}(i_1) \Phi^*_{k_1}(i_1) \cdots \Phi_{k_n}(i_n)
\Phi^*_{k_n}(i_n) \right)}{\tr(u^{H_s} t^{C_s}) }.
\end{equation}
with
\begin{equation}
  \Phi_k(i) := \Ad \left( \Gamma_+^s(\rho_i^\rightarrow) \Gamma_-^s(\rho_i^\leftarrow)^{-1} \right) \cdot \phi_k, \qquad \Phi^*_k(i) := \Ad \left( \Gamma^s_+(\rho_i^\rightarrow) \Gamma^s_-(\rho_i^\leftarrow)^{-1} \right) \cdot \phi^*_k
\end{equation}
where $\Ad$ denotes the Lie group adjoint action $\Ad(A) \cdot B := A B A^{-1}$ and where
\begin{equation}
  \label{eq:rhoarrowdef_S}
  \rho_i^\rightarrow := \bigcup_{\ell=1}^i \rho_i^+ \cup \bigcup_{n \geq 1} u^{n} \rho^+, \qquad
  \rho_i^\leftarrow := \bigcup_{\ell=i}^N \rho_i^- \cup \bigcup_{n \geq 1} u^{n} \rho^-.
\end{equation}

Each $\Phi_k(i)$ and $\Phi^*_k(i)$ is a linear combination of the $\phi_k$'s (recall that $\phi^*_k = (-1)^k \phi_{-k}$ is just a relabelling of $\phi_k$ added for convenience), a statement which follows easily if one passes to the generating field $\phi(z)$ and uses the commutation between $\Gamma^s_{\pm}$ and $\phi(z)$ given in~\eqref{eq:gampsi_S}. By Wick's lemma~\ref{lem:wickgen_S}, the left-hand side of equation~\eqref{eq:rhotdef_S} is a $2n \times 2n$ pfaffian of a certain antisymmetric matrix $K$ whose entries $K_{\gamma, \delta}$ for $\gamma < \delta$ are
\begin{equation}
\begin{split}
  \begin{cases}
  2^{-1}\tr \left( u^{H_s} t^{C_s} \Phi_{k_{\gamma}} (i_{\gamma}) \Phi_{k_{\delta}} (i_{\delta}) \right) / \tr (u^{H_s} t^{C_s}), &1 \leq \gamma < \delta \leq n,\\
  2^{-1}\tr \left( u^{H_s} t^{C_s} \Phi_{k_{\gamma}} (i_{\gamma}) \Phi^*_{k_{2n+1-\delta}} (i_{2n+1-\delta}) \right) / \tr (u^{H_s} t^{C_s}), &1 \leq \gamma \leq n < \delta \leq 2n,\\
  2^{-1}\tr \left( u^{H_s} t^{C_s} \Phi^*_{k_{2n+1-\delta}} (i_{2n+1-\delta}) \Phi^*_{k_{2n+1-\gamma}} (i_{2n+1-\gamma}) \right) / \tr (u^{H_s} t^{C_s}), &n < \gamma < \delta \leq 2n.
  \end{cases}
\end{split}
\end{equation} 
Finally one checks that, after passing to the generating field $\phi(z)$ via coefficient extraction, one has
\begin{equation}
  \begin{split}
    &\tr \left( u^{H_s} t^{C_s} \Phi_{k_{\gamma}} (i_{\gamma}) \Phi_{k_{\delta}} (i_{\delta}) \right) = \\
    &\qquad [z^{k_{\gamma}} w^{k_{\delta}}]  \tr \left (u^{H_s} t^{C_s} \Gapl^s (\rho_1^+) \cdots \phi (z) \cdots \phi(w) \cdots \Gami^s (\rho_N^-) \right), \qquad 1 \leq \gamma < \delta \leq n, \\
    & \tr \left( u^{H_s} t^{C_s} \Phi_{k_{\gamma}} (i_{\gamma}) \Phi^*_{k_{2n+1-\delta}} (i_{2n+1-\delta}) \right) = (-1)^{k_{2n+1-\delta}}  \times \\
    & \qquad [z^{k_{\gamma}} w^{-k_{2n+1-\delta}}] \tr \left (u^{H_s} t^{C_s} \Gapl^s (\rho_1^+) \cdots \phi (z) \cdots \phi(w) \cdots \Gami^s (\rho_N^-) \right), \qquad 1 \leq \gamma \leq n < \delta \leq 2n, \\
    & \tr \left( u^{H_s} t^{C_s} \Phi^*_{k_{2n+1-\delta}} (i_{2n+1-\delta}) \Phi^*_{k_{2n+1-\gamma}} (i_{2n+1-\gamma}) \right) = (-1)^{k_{2n+1-\gamma}+k_{2n+1-\delta}}  \times \\ 
    & \qquad [z^{-k_{2n+1-\gamma}} w^{-k_{2n+1-\delta}}] \tr \left (u^{H_s} t^{C_s} \Gapl^s (\rho_1^+) \cdots \phi (z) \cdots \phi(w) \cdots \Gapl^s(\rho_N^+) \right), \qquad n < \gamma < \delta \leq 2n. 
  \end{split}
  \end{equation}
  Note that the first equation above corresponds to $\phi \phi$ correlators, the second to $\phi \phi^*$ while the third to $\phi^* \phi^*$ correlators. This explains the presence of minus signs in the second and third equations as $\phi_{k} = (-1)^k \phi^*_k$. The kernel has been arranged in this way so as to make the $u=t=0$ limit correspond to the kernels of Matsumoto and Vuleti\'c~\cite{mat, vul2}. Finally, each of the above quantities inside a coefficient extraction can be computed in exact form by $\Gamma$-elimination. At the end one is left with $\tr (u^{H_s} t^{C_s} \phi(z) \phi(w) )$ which is given by Proposition~\ref{prop:fermionic_expectations_S}. Putting it all together, one arrives at the stated values for the entries of the correlation kernel. This concludes the proof of Theorem~\ref{thm:main_thm_S}. \qed

  We can also derive a strict analogue of Proposition~\ref{prop:ncont}. Recall that $\langle O \rangle := \tr (u^{H_s} t^{C_s} O) / \tr (u^{H_s} t^{C_s})$. We write the $n$-point correlation function as follows: 
  \begin{equation}
    \begin{split}
    \tilde{\varrho}^s(U) &= 2^{-n} \left[ \prod_{\ell=1}^n z_\ell^{k_\ell} (-1)^{k_\ell} w_\ell^{-k_\ell} \right]  \tr \left( u^{H_s} t^{C_s} \Gapl^s (\rho_1^+) \Gami^s (\rho_1^-) \cdots \Gapl^s (\rho_{i_1}^+) \phi(z_1) \phi(w_1) \Gami^s (\rho_{i_1}^-) \cdots \right.\\
    &\qquad \left. \Gapl^s (\rho_{i_n}^+) \phi(z_n) \phi(w_n) \Gami^s (\rho_{i_n}^-) \cdots \Gapl^s (\rho_N^+) \Gami^s (\rho_N^-) \right)/\tr \left( u^{H_s} t^{C_s} \right) \\
    &= 2^{-n} \left[ \prod_{\ell=1}^n z_\ell^{k_\ell} (-1)^{k_\ell} w_\ell^{-k_\ell} \right] \prod_{\ell=1}^n \left( F(i_\ell, z_\ell) F(i_\ell, w_\ell) \right) \cdot \langle \phi(z_1) \phi(w_1) \cdots \phi(z_n) \phi(w_n) \rangle
    \end{split}
  \end{equation}
where the second equality follows from the first after cyclically moving the $\Gamma^s$ operators out of the way and picking up the corresponding $F$ factors and where brackets stand for coefficient extraction. Using the pfaffian evaluation from Proposition~\ref{prop:pf_strict_eval} and replacing coefficient extraction with contour integration we arrive at the following result.
\begin{prop}
\label{prop:ncont_S}
The $n$-point correlation has the form
\begin{multline}
    \tilde{\varrho}^s(U) = \frac{(1+t)^n (u;u)_{\infty}^{2n}}{2^n (-u;u)_{\infty}^{n}} \oint \cdots \oint \prod_{\ell=1}^n \left( \frac{ d z_\ell d w_\ell F(i_\ell, z_\ell) F(i_\ell, w_\ell)}{(2 \pi \im)^2 (-1)^{k_\ell} z^{k_\ell + 1}  w^{-k_\ell + 1}} \cdot \frac{\theta_u \left( \frac{z_\ell}{w_\ell} \right)}{\theta_u \left( -\frac{z_\ell}{w_\ell} \right)} \right) \times \\ \prod_{1 \leq \ell < m \leq n} \frac{\theta_u \left( \frac{z_\ell}{z_m} \right) \theta_u \left( \frac{w_\ell}{w_m} \right) \theta_u \left( \frac{z_\ell}{w_m} \right) \theta_u \left( \frac{w_\ell}{z_m} \right)} {\theta_u \left( -\frac{z_\ell}{z_m} \right) \theta_u \left( -\frac{w_\ell}{w_m} \right) \theta_u \left( -\frac{z_\ell}{w_m} \right) \theta_u \left( -\frac{w_\ell}{z_m} \right)}
\end{multline}
where the conditions on the contours are the same as in Proposition~\ref{prop:ncont}.
\end{prop}

\section{Conclusion} \label{sec:conclusion}

We conclude with a few remarks about how this work fits in the bigger
picture. Let us first mention that this paper is part of a project to
analyze Schur processes with ``nonstandard'' boundary conditions. In
another paper written in collaboration with Peter Nejjar and Mirjana
Vuleti\'c, we introduced the Schur process with free
boundaries~\cite{bbnv}. For the edge behavior of this process, we
expect to encounter pfaffian variants of the finite-temperature Airy
process~\cite{bbnv2}. Note that this is a priori unrelated to the
pfaffian process considered in Section~\ref{sec:corr_S}, which we
discuss below.

As said in Remark~\ref{rem:universality}, the finite-temperature Airy
kernel should be the universal scaling limit of the periodic Schur
process around a generic point at the edge of the limit shape. It is
however known that, in the zero-temperature case, there exist other
scaling limits (e.g., around nongeneric points such as cusps) and other specializations (e.g., those of $z$-measures) which lead to
interesting objects like the Pearcey kernel~\cite{or2}, the cusp Airy kernel~\cite{or3}, the
hypergeometric kernel~\cite{bo2}, etc. It is natural to ask whether these kernels have
finite-temperature variants.

We have defined in Section~\ref{sec:cylPlanchproc} the stationary
cylindric Plancherel process through its finite-dimensional marginals,
but we believe that it admits alternative nice
combinatorial/probabilistic constructions via variants of the
Robinson-Schensted correspondence (as used in~\cite{BO06} for the
zero-temperature limit), Hammersley's process~\cite{AD95} or the
polynuclear growth (PNG) model~\cite{PS02}. We might return to this
subject in another publication.

On a related aspect, it would be nice to devise a sampling
algorithm---for the stationary cylindric Plancherel process or for the
periodic Schur process in general---as was done in
\cite{bordyn,BBBCCV} for the zero-temperature case. In fact, we
believe that the approach of~\cite{BBBCCV} may be adapted
straightforwardly (but we have not implemented it yet), building on
ideas present in~\cite{Langer,bcc}.

Our characterization in Section~\ref{sec:corr_S} of the correlation
functions of the periodic strict Schur process raises the question
whether they admit new (bulk or edge) scaling limits. For the edge, we
expect to still get generically the finite-temperature Airy process
(pfaffians reducing to determinants in the limit). Bulk limits might
be more interesting and involve pfaffian variants of the kernels
obtained in~\cite{B2007cyl}. See~\cite{vul2} for the zero-temperature
case.

As said in the introduction, $M_{\alpha}$ has first appeared in edge
asymptotics of the Moshe--Neuberger--Shapiro (MNS) matrix
model~\cite{joh3}, 1D free fermions in a harmonic
trap~\cite{DLMS15,DLMS16,lw} and nonintersecting Ornstein--Uhlenbeck
processes on a cylinder~\cite{dms}. All these models are different
facets of the same continuous space-time finite-temperature
free-fermionic model. Let us point out that the free fermions
used here are slightly different from those considered in the above
references. They correspond essentially to an effective theory when
the number of fermions gets large.

Perhaps coincidentally, up to reparametrization, $M_{\alpha}$ has also
appeared in non-free-fermionic models and at the level of the KPZ
equation itself, as can be seen in the works of Amir--Corwin--Quastel
and Sasamoto--Spohn (see~\cite{acq, ss}, and references therein) where
it appears in the scaling of the weakly asymmetric exclusion process,
and in works of Calabrese--Le Doussal--Rosso, Borodin--Corwin--Ferrari
and Imamura--Sasamoto~\cite{cdr, bcf, is} regarding fluctuations of
the free energy of the O'Connell--Yor directed polymer. This
coincidence was noticed in~\cite{DLMS15} and further discussed
in~\cite{DLMS16,lw,cg} but, to the best of our knowledge, its
fundamental origin is not understood.

Finally, we note that the approach of Baik, Deift and
Johansson~\cite{bdj} for the edge behavior of Plancherel random
partitions was based on Riemann-Hilbert techniques, while the approach
that we follow here is much closer in spirit to that developed by
Borodin, Okounkov and Olshanski~\cite{boo}. It seems an interesting
technical challenge to adapt the Riemann-Hilbert approach to analyze
the Fredholm determinant $F_\alpha$, even though its asymptotics may
be studied through other techniques---see~\cite{cg} and references
therein. We might return to this question in the future.

\appendix

\section{Eta and theta functions}
\label{sec:theta}

Fix $q$ a complex parameter of modulus less than 1. The $q$-Pochhammer
symbol of argument $z$ and length $n \in \N \cup \{ \infty \}$, the
\emph{multiplicative} theta function and the Dedekind eta function are respectively
defined as
\begin{equation}
  (z;q)_n := \prod_{k =0}^n (1-z q^k), \qquad \theta_q(z):= (z; q)_{\infty} (q/z; q)_{\infty}, \qquad \eta(q) := q^{\frac{1}{24}} (q;q)_\infty.
\end{equation}
We also introduce the (``additive'') Jacobi theta functions $\theta_3$
and $\theta_1$:
\begin{equation}
  \label{eq:theta31def}
  \theta_3 (z; q) := \sum_{n \in \Z} q^{\frac{n^2}{2}} z^n, \qquad
  \theta_1 (z; q) := \frac{1}{\im} \sum_{n \in \Z} (-1)^n q^{\frac{(n+1/2)^2}{2}} z^{n+\frac{1}{2}} = \frac{q^{\frac{1}{8}} z^{\frac{1}{2}}}{\im} \theta_3 (-q^{1/2} z; q) .
\end{equation}
These conventions differ by the change $q \to q^{1/2}$ from those in
\cite[p.355]{Bateman} and that for $\theta_1$ is also slightly
different from \cite[p.411]{B2007cyl} (the $\im$ factor makes
$\theta_1(z;q)$ real for $q$ real and $|z|=1$).
The Jacobi triple product identity can then be written in the two
equivalent forms
\begin{equation} \label{eq:jtp}
  \theta_3(z;q) = (q;q)_\infty \theta_q(-q^{1/2} z), \qquad
  \theta_1(z;q) = \im q^{\frac{1}{8}} z^{-\frac{1}{2}} 
  (q;q)_\infty \theta_q (z)
\end{equation}
where in the last equation we have used the relation $\theta_q(qz) = \theta_q(1/z) = -(1/z) \theta_q(z)$.

\section{Basics of charged free fermions}
\label{sec:fermrem}

We recall the definitions necessary for
Section~\ref{sec:fermions}---see also \cite[Section~3]{bbnv} and
references therein. The \emph{fermionic Fock space}, denoted $\F$, may
be seen as the infinite dimensional Hilbert space with orthonormal
basis indexed by Maya diagrams, as defined in
Section~\ref{sec:fermidirac}. We use the bra--ket notation throughout,
hence denote respectively by $\bra{\underline{n}}$ and
$\ket{\underline{n}}$ the bra and ket associated with the Maya diagram
$\underline{n}$. The charge $C$ and energy $H$, as defined in
\eqref{eq:CHdef}, are naturally promoted as diagonal operators on
$\F$. We also defined the shift operator $R$, that shifts a Maya diagram
one unit to the right.

\paragraph{Creation/annihilation operators.}
For $k \in \Z'$, the fermionic \emph{creation and annihilation operators} $\psi_k$ and $\psi^*_k$ are defined by
\begin{equation}
  \label{eq:psidef}
  \psi_k \ket{\underline{n}} :=
  \begin{cases}
    0 & \text{if $n_k=1$}, \\
    (-1)^{\sum_{j>k} n_j} \ket{\underline{n}^{(k)}}, & \text{if $n_k=0$}, \\
  \end{cases} \qquad
  \psi_k^* \ket{\underline{n}} :=
  \begin{cases}
    (-1)^{\sum_{j>k} n_j} \ket{\underline{n}^{(k)}},  & \text{if $n_k=1$}, \\
    0 & \text{if $n_k=0$}, \\
  \end{cases}
\end{equation}
where $\underline{n}^{(k)}$ denotes the Maya diagram obtained from
$\underline{n}$ by inverting the value found at position $k$. Because
of the signs, we have the \emph{canonical anticommutation relations}
\begin{equation}
  \label{eq:car}
  \{ \psi_k, \psi_\ell^* \} = \delta_{k,\ell}, \qquad
  \{ \psi_k, \psi_\ell \} = \{ \psi_k^*, \psi_\ell^* \} = 0, \qquad
  k,\ell \in \Z'
\end{equation}
where $\{a,b\}:=ab+ba$ denotes the anticommutator. We also define the generating series
\begin{equation}
 \label{eq:psigendef}
 \psi(z) := \sum_{k \in \Z'} \psi_k z^k, \qquad
 \psi^*(w) := \sum_{k \in \Z'} \psi^*_k w^{-k}.
\end{equation}
Note that the charge and energy operators can be rewritten as
bilinears in the creation/annihilation operators, namely
\begin{equation}
  C = \sum_{k \in \Z'} :\psi_k \psi^*_k:,\qquad H = \sum_{k \in \Z'} k :\psi_k \psi^*_k:, \qquad
  :\psi_k \psi^*_k:\ := 
  \begin{cases}
    \psi_k \psi^*_k,& \text{ if } k > 0,\\
    -\psi_k^* \psi_k,& \text{ if } k < 0.
  \end{cases}
\end{equation}
Here it is convenient to introduce the normal ordering $:\cdot:$ with
respect to the vacuum so as to have well-defined single sums.  Such
bilinear quantities are essential to free field theory. A
manifestation of this is Wick's lemma at finite temperature given below.

\paragraph{$\Gamma$-operators.} For $\rho$ a specialization with
generating function $H(\rho,\cdot)$ as defined in \eqref{eq:Hrhodef},
we consider the sequence $(p_n(\rho))_{n \geq 1}$ of \emph{power sums
  specialized at $\rho$} given by
\begin{equation}
  \sum_{n \geq 1} \frac{p_n(\rho)}{n} z^n :=\ln H(\rho;z).
\end{equation}
The \emph{half-vertex operators} $\Gamma_\pm(\rho)$ are then defined by
\begin{equation}
  \label{eq:gamalpdef}
  \Gamma_\pm(\rho) := \exp \left( \sum_{n \geq 1} \frac{p_n(\rho) \alpha_{\pm n}}{n} \right), \qquad
  \alpha_n := \sum_{k \in \Z'} \psi_{k-n} \psi_k^*.
\end{equation}
If $x$ is a variable, we denote by $\Gamma_\pm(x)$ (respectively
$\Gamma'_\pm(x)$) the half-vertex operators for the specialization in
the single variable $x$ (respectively its dual $\bar{x}$), for which
$p_n(x)=x^n$ (respectively $p_n(\bar{x})=(-1)^{n-1} x^n$). It is well-known that
\begin{equation}
  \label{eq:schurelem}
  \bra{\underline{n}} \Gamma_+(\rho) \ket{\underline{m}} =
  \bra{\underline{m}} \Gamma_-(\rho) \ket{\underline{n}} =
  \begin{cases}
    s_{\mu/\nu}(\rho) & \text{if $c=c'$,} \\
    0 & \text{otherwise,}
  \end{cases}
\end{equation}
where $(\nu,c)$ and $(\mu,c')$ are the charged partitions associated
with the Maya diagrams $\underline{n}$ and $\underline{m}$ respectively,
via the correspondence of Section~\ref{sec:fermidirac}. The
$\Gamma$-operators commute with the charge ($C$) and shift ($R$)
operators and satisfy the following quasi-commutation relations
\begin{equation}
  \label{eq:gamcomm}
  \begin{split}
  \Gamma_\pm(\rho) \psi(z) = H(\rho;z^{\pm 1}) \psi(z) \Gamma_\pm(\rho), &\qquad
  \Gamma_\pm(\rho) \psi^*(w) = H(\rho;w^{\pm 1})^{-1} \psi^*(w) \Gamma_\pm(\rho),\\
  \Gamma_+(\rho) \Gamma_-(\rho') = H(\rho;\rho')
  \Gamma_-(\rho') \Gamma_+(\rho), &\qquad
  H(\rho;\rho') := \exp \left( \sum_{n \geq 1} \frac{p_n(\rho) p_n(\rho')}{n} \right),\\
  \Gamma_{\pm}(\rho) u^H = u^H \Gamma_{\pm}(u^{\pm 1} \rho), &\qquad
  H(u^{\pm 1} \rho;z):= H(\rho;u^{\pm 1} z).
\end{split}
\end{equation}
We may reconstruct the creation/annihilation operators from the
$\Gamma$-operators as
\begin{equation}
  \label{eq:boson_fermion}
  \psi(z) = z^{C-\frac{1}{2}} R\, \Gami(z) \Gatpl\left(-z^{-1}\right),\qquad
  \psi^*(w) = R^{-1} w^{-C+\frac{1}{2}} \Gatmi(-w) \Gapl\left(w^{-1}\right),
\end{equation}
a result often referred to as the \emph{boson--fermion correspondence}.

\paragraph{Wick's lemma.}

For $u \in (0,1)$ and $t > 0$, we denote by $\evut{\cdot}$ the grand
canonical expectation value defined by
$\evut{O} := \tr (t^C u^H O)/\tr( t^C u^H)$ for any operator $O$
acting on Fock space.  Then we have the following
``finite-temperature'' version of Wick's lemma.

\begin{lem}
  \label{lem:wickgen}
  Let $\varPsi$ be the vector space spanned by (possibly infinite)
  linear combinations of the $\psi_k$ and $\psi_k^*$, $k \in \Z'$.
  For $\varphi_1,\ldots,\varphi_{2n} \in \varPsi$, we have
  \begin{equation}
    \label{eq:wickgen}
    \evut{\vf_1 \cdots \vf_{2n}} = \pf A
  \end{equation}
  where $A$ is the $2n \times 2n$ antisymmetric matrix defined by
  $A_{i,j}=\evut{\vf_i \vf_{j}}$ for $i<j$.

  In particular, if $\varphi_{2i-1}$ (respectively $\varphi_{2i}$) is a
  linear combination of the $\psi_k$'s only (respectively the $\psi^*_k$'s
  only) for all $i=1,\ldots,n$, we have
  \begin{equation}
    \label{eq:wickgen_det}
    \evut{\vf_1 \cdots \vf_{2n}} = \det_{1 \leq i,j \leq n} \evut{\mathcal{T}(\vf_{2i-1},\vf_{2j})}
  \end{equation}
  where $\mathcal{T}$ is the ``time-ordered product'':
  $\mathcal{T}(\vf_{2i-1},\vf_{2j})=\vf_{2i-1} \vf_{2j}$ for
  $i \leq j$ and
  $\mathcal{T}(\vf_{2i-1},\vf_{2j})=-\vf_{2j} \vf_{2i-1}$ for
  $i > j$.
\end{lem}

The more usual Wick's lemma at zero temperature corresponds to the
case $u=0$, for which $\evut{\cdot}$ reduces to the vacuum
expectation value $\bra{\emptyset} \cdot \ket{\emptyset}$. For
convenience we provide a proof of Lemma~\ref{lem:wickgen}, which seems
basically due to Gaudin~\cite{Gaudin60}. 

\begin{proof} 
  Let us introduce the \emph{density matrix}
  $D:=t^C u^H / \tr(t^C u^H)$ so that $\evut{O}=\tr(D O)$ for any
  operator $O$.  By direct computation we have
\begin{equation}
  \label{eq:rho_phi_comm}
  D \psi_k = t u^k \psi_k D,\qquad D \psi^*_k = (t u^{k})^{-1} \psi^*_k D. 
\end{equation}
By multilinearity, it suffices to prove~\eqref{eq:wickgen} when each
$\vf_i$ is equal to either $\psi_{k_i}$ or $\psi^*_{k_i}$ for some
$k_i$. The left hand side of~\eqref{eq:wickgen} can be telescopically
rewritten as:
\begin{equation}
  \label{eq:trrew}
  \begin{split}
    \tr(D \vf_1 \cdots \vf_{2n}) &= \sum_{i=2}^{2n} (-1)^{i} \tr (D \vf_2 \cdots \vf_{i-1}\{\vf_1, \vf_i \} \vf_{i+1} \cdots \vf_{2n}) - \tr(D \vf_2 \vf_3 \cdots \vf_{2n} \vf_1) \\
    &= \sum_{i=2}^{2n} (-1)^{i} \{ \vf_1, \vf_i \} \tr (D \vf_2 \cdots \vf_{i-1} \vf_{i+1} \cdots \vf_{2n}) - c_1 \tr(D \vf_1 \vf_2 \vf_3 \cdots \vf_{2n}) \\
  \end{split}
\end{equation}
where $c_1 = (t u^{k_1})^{-1}$ (respectively $=t u^{k_1}$) if
$\vf_1 = \psi_{k_1}$ (respectively $=\psi^*_{k_1}$)---to pass to the
second line we use the fact that, by the canonical anticommutation
relations, $\{ \vf_1, \vf_i \}$ is a scalar, and we use cyclicity and
\eqref{eq:rho_phi_comm} to rewrite the rightmost trace.
In particular, for $n=1$, \eqref{eq:trrew} yields
\begin{equation}
  \label{eq:anticomcor}
 \{\vf_1, \vf_2\} = (1+c_1) \evut{\vf_1 \vf_2}
\end{equation}
which of course still holds when replacing $\vf_2$ by $\vf_i$ for any $i$.
We deduce the recursion relation
\begin{equation}
  \label{eq:pfrec}
  \evut{\vf_1 \cdots \vf_{2n}} = \sum_{i=2}^{2n} (-1)^{i} \evut{\vf_1 \vf_i}
  \evut{\vf_2 \cdots \vf_{i-1} \vf_{i+1} \cdots \vf_{2n}}.
\end{equation}
The proof of \eqref{eq:wickgen} is then done by induction: it is a
tautology for $n=1$, and assuming that it holds at rank $n-1$, we have
$\evut{\vf_2 \cdots \vf_{i-1} \vf_{i+1} \cdots \vf_{2n}}=\pf A^{(1,
  i)}$ where $A^{(1, i)}$ is the $(2n-2)\times (2n-2)$ submatrix of
$A$ with the first and $i$-th rows and columns removed. We recognize
in the right-hand side of \eqref{eq:pfrec} the expansion of the
pfaffian of $A$ with respect to the first row/column, and conclude
that \eqref{eq:wickgen} holds at rank~$n$.

In the case where $\varphi_{2i-1}$ (respectively $\varphi_{2i}$) is a linear
combination of the $\psi_k$'s only (respectively the $\psi^*_k$'s only) for
all $i=1,\ldots,n$, $A_{i,j}$ vanishes whenever $i$ and $j$ have the
same parity, hence $\pf A = \det_{1 \leq i,j \leq n} A_{2i-1,2j}$
which is equivalent to the stated form \eqref{eq:wickgen_det}.
\end{proof}

\begin{rem} The main ingredients of the proof are the facts that (i)
  the anti-commutator of any two elements of $\Psi$ is a scalar, and
  (ii) the density matrix $D$ is the exponential of a bilinear
  combination of elements of $\Psi$. Here we assume that $D$ has a
  diagonal form, which simplifies the proof but is not necessary---see 
  e.g.~\cite[Ch.~4 and P4.1]{BlaizotRipka}. The canonical density
  matrix $\Pi_0 u^H$ does not satisfy the property (ii), which is why
  we need to pass to the grand canonical ensemble to have
  determinantal correlations.
\end{rem}

\section{Basics of neutral free fermions} 
\label{sec:fermions_S}

We recall the theory of neutral free fermions, useful for the study of
the strict Schur process in \cite{vul2}, following the conventions of
\cite[Sections 1.4 and 3.8]{whe} (note they differ by a factor of 2
from the conventions of \cite[Exercise 14.13]{kac} and by a factor of
4 from those of \cite{mat, vul}). We give a more detailed outline---neutral fermions being less often used in the literature---than the one in Appendix~\ref{sec:fermrem} on which we rely in our construction of neutral fermions.

\paragraph{Neutral Fock space and fermionic operators.}

We begin by constructing the \emph{neutral fermionic Fock space} $\NF$
in analogy with how we constructed $\F$.  We start with the
\emph{neutral fermionic operators}, which we define in terms of the
charged fermionic operators $\psi_k, \psi^*_k$ as follows:
\begin{equation}
  \phi_i = \psi_{i+\frac{1}{2}} + (-1)^i \psi^*_{-i+\frac{1}{2}}, \qquad i \in \Z
\end{equation}
(notice how we have switched the indexing from $\Z'$ to $\Z$). They form the \emph{neutral fermionic field}
\begin{equation}
  \phi(z) = \sum_{n \in \Z} \phi_n z^n
\end{equation}
and satisfy the following canonical anticommutation relations:
\begin{equation}
  \label{eq:car_S}
  \{\phi_i, \phi_j\} = 2 (-1)^j \delta_{i, -j}
\end{equation}
which in particular implies 
\begin{equation}
  \phi_i^2 = \delta_{i, 0}.
\end{equation}
Let us further define
\begin{equation}
  \phi^*_i = (-1)^i \phi_{-i}, \qquad i > 0.
\end{equation}
Given an ordered subset $\{ \mu_1 > \cdots > \mu_r \geq 0 \} \subset \N$, a basis for $\NF$ (respectively its dual $\NF^*$) is given by states of the form
\begin{equation}
  \phi_{\mu_1} \cdots \phi_{\mu_r} \vv, \qquad \vcv \phi_{\mu_r}^* \cdots \phi_{\mu_1}^*
\end{equation}
which can be naturally identified with pairs $\ket{\mu, c}$ (respectively $\bra{\mu, c}$) where $\mu$ is a \emph{strict partition} (a partition with all its parts distinct) and $c = 0$ if $r$ is even, and 1 otherwise. The number $c \in \Z_2$ can be seen as an even/odd grading $\NF = \NF_0 \oplus \NF_1$, based on whether a basis state comes from an even/odd number of fermions acting on the vacuum. For example, the vectors $\phi_3 \phi_2 \phi_0 \vv =: \ket{(3,2), 1} \in \NF_1, \phi_3 \phi_2 \vv =: \ket{(3,2), 0} \in \NF_0$ correspond to the subsets $\{ 3,2,0 \}, \{ 3,2 \} \subset \N$ and both to the same strict partition $(3,2)$. One can think of basis states in $\NF$ aforedescribed as collections of finitely many particles sitting on the $\N$ lattice.

Note $\phi^*_k$ for $k>0$ is the adjoint of $\phi_k$ under the inner product 
$\langle \lambda, c | \mu, d \rangle = 2^{\ell(\lambda)} \delta_{\lambda, \mu} \delta_{c, d}$ for $\lambda, \mu$ strict partitions and 
$c,d \in \Z_2$. We have $\phi_{-n} \vv = \phi_n^* \vv = 0 = \vcv \phi_{n} = \vcv \phi_{-n}^* \vv$ for all $n > 0$.

We denote $\bra{\lambda^{\vee}, c} = 2^{-\ell(\lambda)} \bra{\lambda, c}$ so that 
$\langle \lambda^{\vee}, c | \mu, d \rangle = \delta_{\lambda, \mu} \delta_{c, d}$.

The \emph{energy operator} $H_s$ and \emph{grading operator} $C_s$ (the analogue of the charge operator $C$) are given by (the subscript standing for $strict$)
\begin{equation}
  \label{eq:chdef_S}
  u^{H_s} \ket{\lambda, c} := u^{|\lambda|} \ket{\lambda, c}, \qquad t^{C_s} \ket{\lambda, c} := t^c \ket{\lambda, c}.
\end{equation}

For $k > 0$ we have 
\begin{equation}
  \frac{1}{2} \phi_k \phi_k^* \ket{\lambda, c} = \mathbbm{1}_{\lambda \text{ has a part of size } k} \ket{\lambda,c}.
\end{equation}

We also have a (finite-temperature) Wick lemma for neutral fermions similar
to Lemma~\ref{lem:wickgen}:
\begin{lem}
  \label{lem:wickgen_S}
  Let $\varPhi$ be the vector space spanned by (possibly infinite linear combinations of) the $\phi_k$'s. For $\varphi_1,\ldots,\varphi_{2n} \in \varPhi$, we have
  \begin{equation}
    \label{eq:wickgen_S} 
    \tr( u^{H_s} t^{C_s} \varphi_1 \cdots \varphi_{2n} ) = \pf A
  \end{equation}
  where $A$ is the antisymmetric matrix defined by 
  $A_{i,j}=\tr \left( u^{H_s} t^{C_s} \vf_i \vf_{j} \right) / \tr (u^{H_s} t^{C_s})$ for $i<j$.
\end{lem}

\begin{rem}
  We remark that above, the trace of an operator $O$ over $\NF$ is understood as
  \begin{equation}
  \tr O = \sum_{\lambda \in \SPar,\ c \in \{0,1\}} \bra{\lambda^{\vee}, c}O\ket{\lambda, c}
  \end{equation} 
  and we have that $\tr (u^{H_s} t^{C_s}) = (1+t) \prod_{i \geq 1} (1+u^i)$.
\end{rem}

\begin{proof}[Proof of Lemma~\ref{lem:wickgen_S}]
  The proof of Lemma~\ref{lem:wickgen} applies mutatis-mutandis except we now use the canonical anticommutation relations for neutral fermions from equation~\eqref{eq:car_S}. Note we still have the simple commutations
  \begin{equation}
    D \phi_k = u^k t \psi_k D, \qquad D \phi^*_k = (u^{k} t)^{-1} \phi^*_k D
  \end{equation}
  for $k > 0$ where $D := u^{H_s} t^{C_s} / \tr (u^{H_s} t^{C_s})$.
\end{proof}

\paragraph{Bosonic and half-vertex operators.}
We can define the so-called \emph{bosonic operators}
$\alpha^s_{\pm n}$ as follows. Fix $n$ a positive \emph{odd} integer.
Then set
\begin{equation}
  \alpha^s_n := \frac{1}{4} \sum_{k \in \Z} (-1)^k \phi_{-k-n} \phi_{k}.
\end{equation}
We have that $\alpha^s_{-n}$ is the adjoint of $\alpha^s_{n}$, that $\alpha^s_n \vv=0$ for $n>0$, and for $n, m \in 2\Z+1$ that they satisfy the following commutation relations
\begin{equation} \label{eq:alphacommut_S}
  [\alpha_n,\alpha_m] = \frac{n}{2} \delta_{n,-m},  \qquad
  [\alpha_n,\phi(z)] = z^n \phi(z).
\end{equation}
For $\rho$ a (strict) specialization of the algebra $\mathsf{SSym}$ we
define the \emph{half-vertex operators} $\Gamma^s_\pm(\rho)$ by
\begin{equation}
  \Gamma^s_\pm(\rho) := \exp \left( \sum_{n \in 2 \N + 1} \frac{2 p_n(\rho) \alpha^s_{\pm n}}{n} \right).
\end{equation}
When $x$ is a variable, we denote by $\Gamma^s_\pm(x)$ the half-vertex operators for the specialization in the
single variable $x$, for which
$p_n(x)=x^n$. $\Gamma^s_-(\rho)$ is the adjoint of $\Gamma^s_+(\rho)$ for any real
$\rho$, and
\begin{equation}
  \label{eq:gamcancel_S}
  \Gamma^s_+(\rho) \vv = \vv, \qquad \vcv \Gamma^s_-(\rho) = \vcv.
\end{equation}
Given two strict specializations $\rho,\rho'$, as
$p_n(\rho \cup \rho')=p_n(\rho) + p_n(\rho')$, we have
\begin{equation}
  \label{eq:gambranch_S}
  \Gamma^s_+(\rho) \Gamma^s_+(\rho') = \Gamma^s_+(\rho \cup \rho') =
  \Gamma^s_+(\rho') \Gamma^s_+(\rho).
\end{equation}
The commutation relations \eqref{eq:alphacommut_S} and the Cauchy
identity \eqref{eq:cauchy_S} imply that
\begin{equation}
  \label{eq:gamcomm_S}
  \Gamma^s_+(\rho) \Gamma^s_-(\rho') = Q(\rho;\rho')
  \Gamma^s_-(\rho') \Gamma^s_+(\rho)
\end{equation}
while
\begin{equation}
  \label{eq:gampsi_S}
  \Gamma^s_\pm(\rho) \phi(z) = Q(\rho;z^{\pm 1}) \phi(z) \Gamma^s_\pm(\rho).
\end{equation}
These latter relations always make sense at a formal level; at an
analytic level they require that the parameter of $Q(\rho;\cdot)$ be
within its disk of convergence. The crucial property of these half-vertex
operators is that skew Schur's $P$ and $Q$ functions arise as their matrix elements (see \cite{mat} or \cite[Sections 1.4.9 and 3.2.8]{whe}):
\begin{equation}
  \label{eq:schurelem_S}
  \begin{split}
  \bra{\lambda^{\vee}, c} \Gamma^s_+(\rho) \ket{\mu, d} = Q_{\mu / \lambda} (\rho) \delta_{c,d},\qquad \bra{\mu^{\vee}, c} \Gamma^s_-(\rho) \ket{\lambda, d} = P_{\mu / \lambda} (\rho) \delta_{c,d}.
  \end{split}
\end{equation}
where $\lambda, \mu$ are strict partitions and $c,d\in \Z_2$. This
results from \eqref{eq:gampsi_S}, Wick's lemma~\ref{lem:wickgen_S} at
$u=0$, and the Jacobi--Trudi-like identity for Schur's $P$ and $Q$
functions we took for their definition---see~\cite{mat},
\cite[Sections 1.4.9 and 3.2.8]{whe} for elementary proofs.

The half-vertex operators $\Gamma^s$ commute with the grading operator $C_s$ and satisfy the following quasi-commutation with the energy operator $H_s$:
\begin{equation}
  \Gamma^s_{\pm}(\rho) u^{H_s} = u^{H_s} \Gamma^s_{\pm} (u^{\pm 1} \rho).
\end{equation}

Finally, the boson--fermion correspondence in this setting
reads~\cite[Exercise 14.13]{kac}:
  \begin{equation} \label{eq:boson_fermion_S}
    \begin{split}  
      \phi(z) &= R_s \Gami^s(z) \Gapl^s \left(-z^{-1}\right), \\
    \end{split}
  \end{equation}
  where $R_s$ satisfies $R_s^2 = 1$, $R_s \vv = \phi_0 \vv$ and $R_s \phi_i = \phi_i R_s$ for $i \ne 0$.

The following averages are useful for our purposes.
\begin{prop}
  \label{prop:fermionic_expectations_S}
We have
\begin{equation}
  \begin{split}
  \vcv \phi(z) \phi(w) \vv &= \frac{z-w}{z+w}, \qquad |w| < |z|,\\
  \tr (u^{H_s} t^{C_s} \phi(z) \phi(w) ) &= \frac{\theta_{u} \left( \frac{w}{z} \right)}{\theta_{u} \left( - \frac{w}{z} \right)} \cdot (1+t) (u; u)_{\infty}^2 (-u; u)_{\infty}^{-1}, \qquad u^{1/2} < |w| < |z| < u^{-1/2}. 
  \end{split}
\end{equation}
\end{prop}

\begin{proof}
  Notice that the first equality is the $u=t=0$ case of the
  second. For the second we use the boson--fermion
  correspondence~\eqref{eq:boson_fermion_S} and $\Gamma$-elimination.
\end{proof}

Using neutral fermions, we can also derive the following pfaffian
evaluation.
\begin{prop}
  \label{prop:pf_strict_eval}
  We have
  \begin{equation}
  \pf_{1 \leq i < j \leq 2n} \frac{\theta_u \left(\frac{x_i}{x_j}\right)} {\theta_u \left(-\frac{x_i}{x_j}\right)} = \prod_{1 \leq i < j \leq 2n} \frac{\theta_u \left(\frac{x_i}{x_j}\right)} {\theta_u \left(-\frac{x_i}{x_j}\right)}.
  \end{equation}
\end{prop}

\begin{proof}
  For an operator $O$ acting on neutral Fock space $\NF$ let us
  introduce, as usual, its expectation to be
  $\langle O \rangle = \tr (u^{H_s} t^{C_s} O) / \tr (u^{H_s}
  t^{C_s})$. By Wick's lemma~\ref{lem:wickgen_S} we have
  \begin{equation}
    \langle \phi(x_1) \phi(x_2) \cdots \phi(x_{2n-1}) \phi(x_{2n}) \rangle = \pf_{1 \leq i < j \leq n} \langle \phi(x_i) \phi(x_j) \rangle.
  \end{equation}
  The entries of the pfaffian are given by
  Proposition~\ref{prop:fermionic_expectations_S}, and the left-hand
  side can also be written as a product of theta functions using the
  boson--fermion correspondence. The result follows upon cancelling all
  diagonal terms appearing on both sides.
\end{proof}

\begin{rem}
  This pfaffian, while simple to write down, has appeared only recently in the (mathematical) literature in the works~\cite[Remark 2.9]{ros2} and~\cite[Lemma 3.1]{ros1} by Rosengren. Remarkably, using an algebraic geometric-type no-go argument, Rains (personal communication with the first author) has proved it was the most general pfaffian evaluation of the form $\pf_{i<j} A_{i,j} = \prod_{i<j} A_{i,j}$. Taking $u \to 0$ leads to a famous pfaffian evaluation of Schur (see, e.g., \cite[Section III.8]{mac}) $\pf_{i<j} \frac{x_i - x_j}{x_i + x_j} = \prod_{i<j} \frac{x_i - x_j}{x_i + x_j}$. 
\end{rem}

\section{On the discrete finite-temperature Bessel kernel}
\label{sec:moreasymp}

We provide some further analysis of the discrete finite-temperature
Bessel kernel of Section~\ref{sec:app}. Our approach relies on the
contour integral representation \eqref{eq:Kcontour} where we set
$u=e^{-r}$ and $t=1$.

\paragraph{A convenient change of variables.} Set
\begin{equation}
  \label{eq:zwchvar}
  z=e^{\frac{\zeta}{2}+\im \phi}, \qquad  w=e^{-\frac{\zeta}{2}+\im \phi}
\end{equation}
so that $z/w=e^{\zeta}$ and $\sqrt{zw}=e^{\im \phi}$, where $\phi$ may
be integrated over the interval $[-\pi,\pi]$ and $\zeta$ over
$\zeta_0 + \im [-\pi,\pi]$ with $\zeta_0 \in (0,r)$. We get
\begin{equation}
  \label{eq:Kreparam}
  K(a,b) = \frac{1}{4\im \pi^2} \int_{-\pi}^{\pi} d\phi e^{\im (b-a) \phi}
  \int_{\zeta_0-\im \pi}^{\zeta_0+\im \pi} d\zeta
  \frac{\kappa(\zeta)}{e^{\frac{a+b}{2} \zeta}} e^{4 L \cos \phi \sinh \frac{\zeta}{2}}.
\end{equation}
where $\kappa(\zeta):=\sum_{m \in \Z'} \frac{e^{\zeta m}}{1+e^{rm}}$.
Note that the modulus of the exponential factor
$e^{4 L \cos \phi \sinh \frac{\zeta}{2}}$ is equal to
$e^{4L \cos \phi \sinh \frac{\zeta_0}{2} \cos \frac{\Im(\zeta)}{2}}$
which is maximal at $(\phi,\Im(\zeta))=(0,0)$, a property which is
essential for doing saddle-point approximations.

\begin{rem}
  Applying the Poisson summation formula as in
  Lemma~\ref{lem:kapbound}, we get
  \begin{equation}
    \label{eq:kappoisson}
    \kappa(\zeta) = \sum_{\ell \in \Z} (-1)^\ell
    \frac{\pi}{r \sin \pi \frac{\zeta-2\im \pi \ell}{r}}.
  \end{equation}
  from which it is manifest that $\kappa$ is a meromorphic function of
  $\zeta$ in $\C$, having a simple pole at each point of the lattice
  $r \Z+2\im \pi \Z$, with residue $(-1)^{n+\ell}$ at
  $\zeta=n r+2 \im \pi \ell$. If we assume that
  $\arg(\zeta) \in [-\pi,\pi]$, then keeping only the term $\ell=0$ in
  \eqref{eq:kappoisson} yields an approximation of $\kappa(\zeta)$
  with uniform error $O(e^{-\pi^2/r}/r)$.
\end{rem}

\begin{rem}
  By integrating over $\phi$ in \eqref{eq:Kreparam}, we get
  \begin{equation}
    \label{eq:Kzeta}
    K(a,b) = \frac{1}{2 \im \pi} \int_{\zeta_0-\im \pi}^{\zeta_0+\im \pi}
    d\zeta \frac{\kappa(\zeta)}{e^{\frac{a+b}{2} \zeta}} I_{b-a}\left(4 L \sinh \frac{\zeta}{2}\right)
  \end{equation}
  where $I_D$ is the modified Bessel function of the first kind and order $D$.
\end{rem}

\paragraph{The bulk limit revisited.} We now give a short rederivation
of Borodin's bulk limiting kernel for the cylindric Plancherel measure.

\begin{prop}[{see \cite[Example~3.4]{B2007cyl}}]
  \label{prop:Borbulker}
  Consider the limit $u=e^{-r} \to 1^-$ with $Lr \to \gamma$ fixed, and assume
  that $ra,rb \to \tau \in \mathbb{R}$ with $D:=b-a$ fixed. Then
  \begin{equation}
    \label{eq:Borbulker}
    K(a,b) \to \frac{1}{2\pi} \int_{-\pi}^{\pi} \frac{e^{\im D \phi} d\phi}{1+e^{\tau-2\gamma\cos \phi}}
  \end{equation}
\end{prop}

\begin{proof}
  We use the integral representation~\eqref{eq:Kreparam} for $K(a,b)$,
  and perform the change of variable $\zeta=\frac{r}{2} + \im r \nu$
  with $\nu \in [-\pi/r,\pi/r]$. Then, using Lemma~\ref{lem:kapbound}
  and simple Taylor expansions, we get
    \begin{equation}
    \begin{split}
      K(a,b) &= \int_{-\pi}^{\pi} \frac{e^{\im D \phi}d\phi}{4\pi^2}
      \int_{-\pi/r}^{\pi/r} d\nu\, e^{(\gamma\cos \phi-\tau/2)(1+2\im
        \nu)+o(1)}
      \left(\frac{\pi}{\cosh \pi \nu} + O(e^{-\pi^2/r})\right)\\
      &\to \int_{-\pi}^{\pi} \frac{e^{\im D \phi}d\phi}{4\pi}
      \int_{-\infty}^{\infty} d\nu \frac{e^{(\gamma \cos
          \phi-\tau/2)(1+2\im \nu)}}{ \cosh \pi \nu},
    \end{split}
  \end{equation}
  where we use standard tails pruning/completion arguments to justify
  passing to the second line.  By integrating over $\nu$ we obtain
  wanted expression~\eqref{eq:Borbulker}.
\end{proof}

We believe that this proof may be adapted---keeping the same change of
variables $(z,w) \to (\phi,\nu)$---to give a streamlined derivation
of~\cite[Theorem~3.1]{B2007cyl} regarding the bulk limiting kernels
for more general periodic Schur processes.

\begin{rem}
  \label{rem:bulk} If we let $\gamma \to \infty$ with $\tau/\gamma=x$
  fixed in the right-hand side of~\eqref{eq:Borbulker}, we
  recover the discrete sine kernel as for the (noncylindric)
  poissonized Plancherel measure---see the remark at the end of
  \cite[Example~3.4]{B2007cyl}. In all rigor,
  Proposition~\ref{prop:Borbulker} does not apply to the situation
  where we let $u \to 1^-$ and $\gamma \to \infty$ jointly as in
  Theorem~\ref{thm:limedge}. But it is not difficult to adapt to the
  finite-temperature setting the proof given
  in~\cite[Section~3.2]{Oko02} for the convergence of the discrete
  Bessel kernel to the sine kernel, the key fact being that
  $\kappa(z,w)$ has a residue $1$ at $z=w$---as mentioned in
  Section~\ref{sec:fermions} this is a general consequence of the
  fermionic canonical anticommutations.
\end{rem}

\paragraph{Proof of Proposition~\ref{prop:Kconvgum}.}

The difficulty to circumvent is that, in the contour integral
representation \eqref{eq:Kcontour} of $K(a,b)$, the exponential factor
$e^{L(z-z^{-1}-w+w^{-1})}$ is now subdominant. Instead, we may see
that the asymptotics is governed by the pole of $\kappa$ at
$z/w=u^{-1}$. The strategy is to deform the integration contours into
two circles such that $|z/w|$ is \emph{larger} than $u^{-1}$ (but
smaller than $u^{-2}$). As we cross the pole, we pick a contribution
from the residue which turns out to be dominant.

More precisely, let us first establish the pointwise
convergence~\eqref{eq:Kconvgum}. We use the
reparametrization~\eqref{eq:Kreparam} and change the integration
contour of $\zeta$ into the segment $\frac{3r}{2}+\im [-\pi,\pi]$ plus
a small negatively oriented circle around the pole at $r$. Noting that
the residue of $\kappa$ at this pole is $-1$ and using the integral
representation
\begin{equation}
  \label{eq:BesselI}
  I_D(G) = \frac{1}{2\pi} \int_{-\pi}^\pi d \phi e^{\im D \phi+G \cos \phi} 
\end{equation}
for the modified Bessel function of the first kind and integer order $D$, we find that
\begin{equation}
  \label{eq:Kedgeexpnew}
  r^{-1} K(a,b) =
  \underbrace{\frac{I_{b-a}(G)}{r e^{\frac{a+b}{2} r}}}_{=:T_1} +
  \underbrace{\frac{1}{4\im \pi^2 r} \int_{-\pi}^{\pi} d\phi e^{\im (b-a) \phi}
  \left( \int_{\frac{3r}{2}-\im \pi}^{\frac{3r}{2}+\im \pi} d\zeta
    \frac{\kappa(\zeta)}{e^{\frac{a+b}{2} \zeta}} e^{4 L \cos \phi \sinh
      \frac{\zeta}{2}}\right)}_{=:T_2}
\end{equation}
where $G:=4L\sinh\frac{r}{2}$. We claim that, in the edge high
temperature regime~\eqref{eq:edgescalgum}, the first term $T_1$
in~\eqref{eq:Kedgeexpnew} tends to $e^{-x} \delta_{x,y}$ (the
Kronecker delta) while the second term $T_2$ tends to $0$.  Note that
$G=2Lr+o(1)\sim 2\gamma$ may or may not remain finite, and the forthcoming
discussion is valid in both cases.

The fact that $T_1 \to e^{-x}$ for $x=y$ (i.e.\,$b-a=0$) is a
straightforward consequence of the
scaling~\eqref{eq:edgescalgum}. Indeed we have
\begin{equation}
  \label{eq:emr}
  e^{\frac{a+b}{2} r} \sim \frac{I_0(2Lr)}{r} e^{\frac{x+y}{2}} \sim \frac{I_0(G)}{r} e^{\frac{x+y}{2}}
\end{equation}
where it is safe to replace the argument of $I_0$ even if $G \to \infty$ since
$I_0(G) \sim \frac{e^G}{\sqrt{2\pi G}}$.

For $x \neq y$, the order $D:=b-a$ of the Bessel function now tends to $\pm \infty$ as
$(y-x)r^{-1}$, and showing that $T_1 \to 0$ amounts to showing that
\begin{equation}
  \label{eq:IDGestim}
  I_{D}(G)=o(I_0(G))=o\left(\frac{e^{G}}{\sqrt{G}}\right)
\end{equation}
which is done by the saddle-point method: in the integral
representation~\eqref{eq:BesselI}, let us move the integration path to
the segment $[-\pi,\pi]+\im\, \arsinh(\frac{D}{G})$. On this
segment we may bound the integrand as
\begin{equation}
  |e^{\im D \phi} e^{G \cos \phi}| = e^{-D \arsinh(\frac{D}{G}) +
    G \sqrt{1+\frac{D^2}{G^2}} \cos \Re(\phi)} \leq e^{-G f(\frac{D}{G}) - G \frac{\Re(\phi)^2}{8}}
\end{equation}
where $f(x):=x\, \arsinh(x)-\sqrt{1+x^2}$. By integrating over
$\phi$ we get
$I_D(G)=O\left(G^{-1/2} \exp\left(-G f \left(\frac{D}{G}\right) \right) \right)$. We now
observe that $f'(x)=\arsinh(x)$ so $f$ is convex and minimal at
$x=0$, with $f(x) \geq -1+\frac{x^2}{2}$. Hence
$e^{-G
  f(\frac{D}{G})}=O\left(e^{G-\frac{D^2}{2G}}\right)=O\left(e^{G-\frac{(y-x)^2}{4Lr^3}}\right)=o(e^G)$,
which establishes~\eqref{eq:IDGestim} hence that $T_1 \to 0$ for $x \neq y$.

We now show that $T_2\to 0$: using~\eqref{eq:kappoisson} and
proceeding as for Lemma~\ref{lem:kapbound}, we see that
$|\kappa(\zeta)|$ may be bounded over the $\zeta$-integration contour
by $\frac{\pi}{r \cosh \frac{\pi \Im(\zeta)}{r}}$ up to some uniform
constant. We then bound the exponential factors as
\begin{equation}
  \left\lvert \frac{e^{4 L \cos \phi \sinh
        \frac{\zeta}{2}}}{e^{\frac{a+b}{2} \zeta}} \right\rvert = \frac{e^{4L\sinh \frac{3r}{2} \cos \phi \cos \Im(\zeta)}}{e^{\frac{3(a+b)r}{4}}} =O\left( \frac{e^{\frac{3G}{2} \left(1-\frac{\phi^2}{8}\right)}}{e^{\frac{3(a+b)r}{4}}}\right)=O\left((r\sqrt{G})^{3/2} e^{-\frac{3 G \phi^2}{16}}\right)
\end{equation}
where we use again~\eqref{eq:emr} and the fact that
$I_0(G)=\Theta\left(\frac{e^G}{\sqrt{G}}\right)$. Upon integrating over $\phi$
and $\zeta$, and multiplying by the prefactor $r^{-1}$, we deduce
that $T_2=O\left( (r \sqrt{G})^{1/2}\right) \to 0$, which
completes the proof of the pointwise convergence~\eqref{eq:Kconvgum}.

The proof of the trace convergence~\eqref{eq:Ktrconvgum} is entirely
similar: starting from~\eqref{eq:Kedgeexpnew}, we find that
\begin{equation}
  \label{eq:Kedgeexpnewtr}
  \sum_{i \geq m} K(i,i) =
  \frac{I_0(G)}{(1-e^{-r}) e^{mr}} +
  \frac{1}{4\im \pi^2} \int_{-\pi}^{\pi} d\phi
  \left( \int_{\frac{3r}{2}-\im \pi}^{\frac{3r}{2}+\im \pi} \frac{d\zeta}{1-e^{-\zeta}}
    \frac{\kappa(\zeta)}{e^{m \zeta}} e^{4 L \cos \phi \sinh
      \frac{\zeta}{2}}\right)
\end{equation}
and all the above analysis of~\eqref{eq:Kedgeexpnew} may be adapted
straightforwardly.
\qed

\bibliographystyle{myhalpha}
\bibliography{cylindric_schur}

\end{document}